\documentclass[a4paper,12pt]{article}
\usepackage{graphicx}
\usepackage{amsmath}
\usepackage{amssymb}
\usepackage{bm}
\usepackage[]{amsthm}
\usepackage[backend=biber,
style=bwl-FU,
sorting=nyt
]{biblatex}
\DeclareDelimFormat{postnotedelim}{\addcomma\space}
\usepackage[labelfont=bf]{caption}
\usepackage{setspace}
\usepackage{amsthm}
\newtheoremstyle{mythm}%
{3pt}
{3pt}
{}
{}
{\bfseries}
{}
{.5em}
{}%
\theoremstyle{mythm}
\newtheorem{remarkex}{Remark}[section]
\newenvironment{remark}
  {\pushQED{\qed}\renewcommand{\qedsymbol}{}\remarkex}
  {\popQED\endremarkex}
\usepackage[left=1.25in, right=1.25in, top=1.25in]{geometry}
\newtheoremstyle{tm}%
{3pt}
{3pt}
{\itshape}
{}
{\bfseries}
{}
{.5em}
{}%
\theoremstyle{tm}
\newtheorem{theorem}{Theorem}[section]
\newtheorem{corollary}{Corollary}[section]
\newtheorem{assumption}{Assumption}
\newtheorem{lemma}{Lemma}[section]
\theoremstyle{definition}
\newtheorem{definition}{Definition}[section]
\newtheorem{proposition}{Proposition}[section]
\newtheorem{example}{Example}[section]
\numberwithin{equation}{section}
\usepackage{hyperref}
\renewcommand{\qedsymbol}{\blacksquare}

\title{Asymptotics of an Explosive Autoregression under Dependence}
\author{Kasper Sunn Blumensaat \\
\small{Department of Economics, University of Oxford} \\
\small{Oxford, United Kingdom}}
\begin{document}
\maketitle 

\begin{abstract}
    We generalize the convergence results of an explosive autoregression, pioneered in \cite{And1959}, in three ways: First, we demonstrate that the centered least-squares estimator converges geometrically to a ratio of limits, even in settings where the innovations are correlated and not centered around zero. Secondly, we demonstrate that the requirement of independent innovations in \cite[Theorem 2.3,]{And1959} can be relaxed to $\alpha$-mixing. Third, we provide an autocorrelation-robust feasible test statistic for the explosive parameter under Gaussian ARMA innovations.
\end{abstract}


\begingroup
\renewcommand{\thefootnote}{}
\footnotetext{The author gratefully acknowledges Bent Nielsen, for his many insightful comments and guidance.}
\endgroup
\newpage

\setcounter{tocdepth}{2}

\section{Introduction}
We expand the classic results of \cite{And1959} by providing three results. We first demonstrate that the strong convergence result of the centered least-squares estimator to a ratio of a forward and backward average of the innovations holds in settings where the innovations are correlated and non-centered. Second, we demonstrate the requirement of independence in \cite[Theorem 2.3,]{And1959} can be loosened to $\alpha$-mixing. Finally, we provide an autocorrelation-robust feasible test statistic for the explosive parameter $\rho$ of an explosive first-order autoregression with potentially autocorrelated Gaussian ARMA innovations, a natural extension of Anderson's t-statistic originally derived under iid Gaussian innovations. \\ 
\leavevmode \\
\cite{And1959} contains three sets of results for a first-order explosive autoregression. \emph{First}, provided the innovations are uncorrelated, have a bounded second moment and centering around zero, the least squares estimator converges to a ratio of the forward and backward averages of the innovations. \emph{Secondly}, in Theorem 2.3, it is shown that if the innovations are additionally assumed to be independent, the limits of the forward and backward averages, if they exist, are independent. \emph{Third}, under $iid$ Gaussian innovations, in the limit, the forward and backward averages are independent Gaussians, allowing for a t-statistic for the explosive parameter with a limiting Gaussian distribution to be provided.
\\ \leavevmode \\
Anderson's original results remain influential in the purely explosive literature. Because the convergence results do not rely on standard invariance principles, Anderson's asymptotic results in the $iid$ Gaussian setting are still frequently utilized when deriving the asymptotic distribution of estimators and tests in an explosive environment. For instance \cite{FulHasGoe1981} and \cite{Jeg1988} generalize Anderson's Gaussian results into a multivariate setting. Additionally, the Gaussian results have seen more recent use in \emph{explosively co-integrated systems}, for instance in Theorem 2.2 of \cite{PhiMag08} and \emph{co-explosive systems}, see Corollary 1 of \cite{Nie2010}. \\ \leavevmode \\
General consistency results, which permit explosivity, but crucially does not require characterization or existence of a limiting distribution of the least squares estimator, have seen advances beyond the $iid$ Gaussian setting. In those cases, the weaker \emph{Marcinkiewicz-Zygmund} condition, provided in Equation (\ref{eq:mz}), suffices. This condition is generally considered quite strong outside of the explosive setting, as it may exclude stochastic volatility and ARCH models. In this way, \cite{LaiWei1983a} provide consistency results for an AR(p), later generalized to VAR(p) in \cite{LaiWei1985}, subsequently generalized to VAR(p) with deterministics in  \cite{Nie2005}. \\
\\
\noindent We divide the results into two main sections: Section \ref{sec:econometric_theory} generalizes the convergence results of the least-squares estimator to forward and backward averages of the innovations. The section also provides sufficient conditions for the forward average to be non-zero \emph{almost surely}, the backward average to converge in distribution, as well as the generalization of \cite[Theorem 2.3]{And1959}. Section \ref{sec:imp} then derives the exact limiting distribution of the least-squares estimator and t-statistic under Gaussian ARMA innovations and provides a feasible autocorrelation-robust t-statistic. This section also shows how the univariate results can be expanded to higher order autoregressions by providing an example of how the AR(1) results may be used in AR(2) with intercept.

\section{Results} \label{sec:econometric_theory}

\subsection{Definitions} \label{sec:def}
The examined data generating process is
\begin{equation} \label{eq:model}
    x_t = \rho x_{t-1} + \epsilon_t .
\end{equation}
(\ref{eq:model}) is a univariate autoregressive explosive process of order 1, with no intercept or deterministic terms.
The autoregressive coefficient $\rho$ is assumed to be greater than one in absolute value. The innovation process $(\epsilon_t)_{t \in \mathbb{N}}$ and initial value $x_0$ are real-valued and defined on a common probability space $(\Omega, \mathcal{F}, \mathbb{P})$. No additional structure is imposed at this stage; intertemporal structure will be introduced explicitly when needed. 
\\ \leavevmode \\
The main asymptotic objects of interest are
\begin{equation*}
    \hat{\rho} := \frac{\sum_{t=1}^T x_t x_{t-1}}{ \sum_{t=1}^T x_{t-1}^2 }, \quad 
    \hat{\rho} - \rho = \frac{\sum_{t=1}^T \epsilon_t x_{t-1}}{ \sum_{t=1}^T x_{t-1}^2 } =: \frac{A_T}{B_T} \quad  \text{ and } \quad \frac{A_T}{\sqrt{B_T}}.
\end{equation*}
The first is the least-squares estimator of a regression of $x_t$ on $x_{t-1}$. $A_T$ and $B_T$ denote the numerator and denominator of the least squares estimator $\hat{\rho}$, centered around its true value $\rho$. And $A_T/ \sqrt{B_T}$ is a t-statistic, originally introduced in \cite{And1959}. \\
 \leavevmode \\
The notation $A_T, B_T$ is used to match notation defined in \cite{And1959}. Following that framework, introduce $ \beta := \rho^{-1}$, which by definition lies in $(-1,1)$, allowing for geometric series such as $\sum_{t=1}^T \beta^t$ to converge, a property which will be relied upon heavily in subsequent convergence results. \\ \leavevmode
\\
Additionally, introduce 
\begin{align}
    z_t :=& \beta^{t-2} x_{t-1} = \begin{cases}
        \rho x_0 \hspace{82pt} \text{ if }t=1  \\ \sum_{i=1}^{t-1} \beta^{i-1} \epsilon_{i} + \rho x_0 \label{eq:zt_def} \hspace{9 pt} \text{ if } t > 1  
    \end{cases},\\
    y_t :=& \sum_{i=1}^t \beta^{t-i} \epsilon_i. \label{eq:y_def}
\end{align}
These two objects are geometrically decaying sums of innovations.  $z_t$ is a \emph{forward average}, attributing highest weight to the \emph{first} innovations, whilst $y_t$ is a \emph{backward average}. \\ \leavevmode
\\
The objects $A_T, B_T, z_T, y_T$ are measurable as they are continuous functions of $(x_0, \epsilon_1, ..., \epsilon_T)$; see Theorem 13.3 of \cite{Bil79}. Throughout the rest of the paper, unless otherwise stated, all convergence occurs on $(\Omega, \mathcal{F}, \mathbb{P})$. The notation used is as follows: $\overset{a.s.}{\rightarrow}$ denotes almost sure convergence, $\overset{\mathbb{P}}{\rightarrow}$ convergence in probability, $\overset{D}{\rightarrow}$ convergence in distribution. Recall that $\overset{a.s.}{\rightarrow}$ implies $ \overset{\mathbb{P}}{\rightarrow} $, which implies $ \overset{D}{\rightarrow}$. $\mathbb{E}[\cdot]$ denotes the expectations operator with respect to measure $\mathbb{P}$.  Theorems, lemmas, propositions and corollaries are proven in the appendix.

\subsection{Convergence Results} \label{sec:found_results}
In this subsection, we introduce Assumption \ref{ass:var}, which imposes a mild moment condition on the innovations and initial value. We then show that this assumption implies important marginal behavior of $z_T, A_T, B_T$. 
\begin{assumption} \label{ass:var} The innovations have a \textbf{bounded second moment}, that is: \vspace{-5mm}
\begin{equation}
    \sup_{t \in \mathbb{N}} \mathbb{E} [\epsilon_t^{2}] \leq M^2< \infty, \quad M \in \mathbb{R}_+.
\end{equation}
Additionally, let the initial value $x_0$ satisfy $\mathbb{E} [x_0^2] \leq M^2$.
\end{assumption}
\begin{remark}
    Assumption \ref{ass:var} relaxes the assumptions in \cite{And1959}, where it is assumed that the innovations are: (1) \emph{uncorrelated} (2) have \emph{expectation zero} (3) have a \emph{constant, bounded second moment} (4) The  initial value $x_0$ is a \emph{finite constant}. ARCH/Stochastic volatility innovations generally satisfy these stronger conditions, as they are uncorrelated, whilst ARMA innovations do not, as these are serially correlated. Assumption \ref{ass:var} permits serial correlation, and therefore permits ARMA innovations.
\end{remark}
\noindent Explosive processes have the unique property of converging towards limiting objects at geometric rates, as $x_T$ itself grows at rate $\rho^T$. Terms containing $x_{T-1}$, such as the numerator $A_T$, grow at rate $\rho^T$, whilst terms like the denominator $B_T$, which contain $x_{T-1}^2$, grow at rate $\rho^{2T}$. One therefore expects the OLS estimator to converge at rate $\rho^T$, and subsequent results will confirm this. \\ 
\\ \leavevmode
One of our central results, is realizing that in an explosive setting, previously derived moment bounds retain their order even when allowing for correlation and non-centering of the innovations. The reason for this is the geometric convergence rate in the explosive settings dominates the quadratic growth of cross-terms. This insight is formalized in bounds of $z_T$, which are provided in Lemmas \ref{lem:z_T_bounds} and \ref{lem:z_bounds} in the Appendix. Once the bounds are established, the same arguments as in \cite{And1959} still apply; one simply replaces the old bounds with the new ones. \\ 
\\ \leavevmode
\noindent To formalize this insight, begin by defining a candidate probability limit $z$:
\begin{equation}
    z := \begin{cases} \label{eq:zdef}
            \lim_{T \rightarrow \infty} z_T \quad \text{ on }  A := \{ \omega: \lim_{T \rightarrow \infty} z_T \text{ exists in }\mathbb{R} \} \\
            0 \hspace{59pt}  \text{ on } A^c
        \end{cases}
\end{equation}
\begin{remark}
    This definition ensures $z$ is a random variable (taking values in $\mathbb{R}$) by giving it an arbitrary value in $\mathbb{R}$ (here $0$) on sets where $z_T$ diverges. This is a standard trick in probability theory; see for instance the proof of Proposition 2.43 of \cite{Brei68}.
\end{remark}
\noindent With definitions clarified, we are now able to state our first result.

\begin{lemma} \label{lem:measurability}
    Under Assumption \ref{ass:var}, $z_T \overset{a.s.}{\rightarrow} z$, with $\mathbb{E} z^2 <\infty$.
\end{lemma}
\noindent $z$ is of large importance, as the next lemma will demonstrate that the OLS estimator converges to functions of this variable.
\begin{lemma}
\label{lem:B_T}
Under Assumption \ref{ass:var}
\begin{align} 
       \beta^{2(T-2)} B_T \hspace{10pt} \overset{a.s.}{\rightarrow}&
        \hspace{3pt} \frac{1}{1-\beta^2} z^2, \label{eq:den_conv} \\
    ( \beta^{T-2} A_T - y_T z ) \overset{a.s.}{\rightarrow}&  \hspace{17pt} \label{eq:num_conv} 0.
\end{align}
\end{lemma}

\noindent  Note that $y_T$ and $z$ are random variables with distributions in most contexts. Despite having characterized the marginal behavior of the numerator and denominator, without further assumptions two issues will arise when we try to combine these marginal results to examine the OLS estimator and the t-statistic. Firstly, there is nothing ensuring that the denominator is non-zero with probability 1. Secondly, the convergence of the distribution $y_T$ is not guaranteed. For the moment, we assume two issues away and explore the consequences of their absence. Section \ref{sec:yz} then examines these assumptions, and where useful, provides low-level sufficient conditions that imply their high-level counterparts.\\ 
\leavevmode \\
\noindent We assume the random variable $z$ is non-atomic at $0$, then explore the consequences.
\begin{assumption} \label{ass:z_0}
    $\mathbb{P}(z \in \{0\})=0$.
\end{assumption}
\begin{remark}
    In subsection \ref{sec:z_0}, we show that Assumption \ref{ass:z_0} is satisfied if the innovations follow conventional Gaussian ARMA innovations, and provide a sufficient condition that may be used to verify whether this condition is satisfied under other innovations.
\end{remark}

\noindent We can now apply the continuous mapping theorem on our previous results without dividing by zero. This allows us to state our first theorem.
\begin{theorem} \label{thm:consistency}
    Under Assumptions \ref{ass:var}-\ref{ass:z_0}
    \begin{align}
        \left| \frac{\rho^T}{\rho^2-1} (\hat{\rho} - \rho) - \frac{y_T}{z} \right| & \overset{a.s.}{\rightarrow} \, 0 \hspace{65pt} \label{eq:white_dist_noy} \\
        \left| \frac{A_T}{\sqrt{B_T}} - Sign(z)\sqrt{1-\beta^2} y_T  \right| & \overset{a.s.}{\rightarrow}  \, 0. \hspace{65pt} \label{eq:and_est_noy} 
    \intertext{Let $(f_T)$ be a sequence converging to 0, then}
        f_T\rho^{T-2} (\hat{\rho}-\rho) &\overset{\mathbb{P}}{\rightarrow} 0. \label{eq:OLS_consistency}
    \end{align}
    Additionally, if $(f_T)$ satisfies $\sum_{t=0}^\infty |f_t| < \infty$, (\ref{eq:OLS_consistency}) converges almost surely.
\end{theorem}
\begin{corollary} \label{cor:consistency}
    Provided Assumptions \ref{ass:var}-\ref{ass:z_0}, the OLS estimator is strongly consistent. That is, $(\hat{\rho}-\rho) \overset{a.s.}{\rightarrow} 0$.
\end{corollary}
\noindent We see that Assumption \ref{ass:z_0} is sufficient to ensure consistency as well as a preliminary strong convergence result on the scaled OLS estimator and the t-statistic. However, as nothing ensures $y_T$ converges, let alone joint convergence of $(z_T, y_T)$, we are not able to put a limiting object on the right-hand-side of the convergence results. In order to improve the rate in (\ref{eq:OLS_consistency}) we must have joint convergence of $(y_T, z_T)$. We assume this and explore its consequences.
\begin{assumption} \label{ass:y_exists}
$(z_T, y_T) \overset{D}{\rightarrow} (z, y)$.
\end{assumption}
\noindent With this high-level assumption in place, we can establish joint convergence of the numerator terms by combining the $a.s.$ convergence results of Theorem \ref{lem:measurability} with the joint convergence of $(y_T, z_T)$ via the Continuous Mapping Theorem.
\begin{lemma} \label{lem:numden_joint}
    Under Assumptions \ref{ass:var}, \ref{ass:y_exists}
    \begin{equation*}
        \left( \beta^{T-2} A_T, (1-\beta^2) \beta^{2(T-2)} B_T \right) \overset{D}{\rightarrow} (yz,z^2).
    \end{equation*}
\end{lemma}

\noindent And finally, combining all three assumptions, we provide the asymptotic behavior of the scaled OLS estimator and the t-statistic $A_T/\sqrt{B_T}$.

\begin{theorem} \label{thm:and_dist}
    Under Assumptions \ref{ass:var}-\ref{ass:y_exists}, jointly
    \begin{align}
        \frac{\rho^T}{\rho^2-1}  (\hat{\rho}-\rho) \quad & \overset{D}{\rightarrow} \quad \frac{y}{z} \label{eq:white_dist} \\ 
        \frac{A_T}{\sqrt{B_T}} \quad &\overset{D}{\rightarrow} \quad Sign(z) \sqrt{ 1 - \beta^2 } \, y.         \label{eq:and_est}
    \end{align}
\end{theorem}
\begin{remark}
    Note that Lemma \ref{lem:measurability} and Theorems \ref{thm:consistency}-\ref{thm:and_dist} do not require \emph{centering} or \emph{uncorrelatedness}, as is done in \cite{And1959}. Therefore, these theorems apply for Gaussian ARMA innovations, a class of innovations which were previously ruled out  in \cite{And1959}. 
\end{remark}
\noindent If further structure is imposed, the $Sign(z)$ term in (\ref{eq:and_est}) can be ignored. 
\begin{assumption} \label{ass:indep}
    $y$ and $z$ are independent.
\end{assumption}
\begin{remark}
    In Subsection \ref{sec:indep}, sufficient conditions for Assumptions \ref{ass:y_exists}-\ref{ass:indep} are provided. These are satisfied for many conventional Gaussian ARMA and latent-state processes. 
\end{remark}
\begin{corollary} \label{cor:symm}
    Under Assumptions \ref{ass:var}, \ref{ass:y_exists} and \ref{ass:indep} and if the innovations satisfy the symmetry condition
    \begin{equation*}
        (\epsilon_1, \epsilon_2, ..., \epsilon_T) \overset{D}{=} - (\epsilon_1, \epsilon_2, ..., \epsilon_T), \quad \forall T \in \mathbb{N}.
    \end{equation*}
    Then, jointly,
    \begin{equation*}
        \frac{A_T}{\sqrt{B_T}} \overset{D}{\rightarrow} \sqrt{1-\beta^2} y, \quad \frac{\rho^T}{\rho^2-1}  (\hat{\rho}-\rho) \overset{D}{\rightarrow} \frac{y}{z}.
    \end{equation*}
\end{corollary}
\noindent This concludes the generalization of Anderson's first results.

\subsection{Properties of $y$ and $z$} \label{sec:yz}
Much of the explosive literature concerned with asymptotic distributions assumes iid Gaussian innovations as this immediately implies Assumptions \ref{ass:var}-\ref{ass:indep}. However, as this paper is concerned with general, potentially dependent innovations, this subsection examines these assumptions and where necessary, provides general low-level sufficient conditions for Assumptions \ref{ass:z_0}-\ref{ass:indep} with an emphasis on ARMA- and stochastic volatility models.

\subsubsection{Sufficient conditions for Assumption \ref{ass:z_0}} \label{sec:z_0}
Intuitively, if the innovations are continuously distributed, $z$, which is a forward average of innovations, ought to be continuously distributed, and so $\mathbb{P}(z \in \{0\})=0$. For independent innovations this is easy to verify, however, when the innovations/initial value are dependent, verification of this becomes more difficult, as one must rule out edge cases where the innovations exactly cancel each other out. \\
\leavevmode \\
\noindent For many latent-state innovations, such as stochastic volatility models, see for instance \cite{She05}, it may be useful to exploit conditional non-atomicity to verify whether $\{z\in \{0\} \}$ is a null set, we do so in the following lemma.
\begin{lemma} \label{lem:SV}
    Suppose the initial value $x_0 \overset{a.s.}{=} 0$ and the innovations $(\epsilon_t)_{t \in \mathbb{N}}$ take on a form $\epsilon_t = \zeta_t \sigma_t$, where $\zeta_t \overset{iid}{\sim}N(0,1)$ with $\sigma_t^2>0$ $a.s.$ for any fixed $t$ and $\sup_{t>0}\mathbb{E}[\sigma_t^2] < \infty$. Furthermore, let $\sigma(\zeta_t: t>0)$ be independent of $\sigma(\sigma_t^2: t >0)$. Then Assumption  \ref{ass:z_0}, stating $\mathbb{P}(z \in \{0\})=0$ is satisfied.
\end{lemma}
\begin{remark}
    This lemma provides an alternative set of sufficient conditions for consistency in the explosive setting where the assumptions made in \cite{Nie2005} are not satisfied. However, the requirement of $\sigma(\zeta_t: t>0)$ being independent of $\sigma(\sigma_t^2: t >0)$ rules out ARCH.
\end{remark}

\noindent For Gaussian ARMA $(\epsilon_t)_{t \in \mathbb{N}}$, one may utilize that $z$ is a sum of Gaussians to obtain a similar result. Before that, let us define the exact class of  Gaussian ARMA processes we will be examining in the following assumption.

\begin{assumption} \label{ass:arma}
    The innovations $(\epsilon_t)_{t \in \mathbb{N}}$ are a stationary Gaussian ARMA(p,q) process with centering around zero. That is for all $t \in \mathbb{Z}$, letting $L$ denote the lag operator, the innovations take on the form
    \begin{equation*}
        \sum_{i=0}^p (1-\varphi_i L)^{i} \epsilon_t = \sum_{i=0}^q (1-\vartheta_i L)^i \zeta_t, \quad \zeta_t \overset{iid}{\sim} N(0,1).
    \end{equation*}
    Additionally, let the roots of the polynomials $\varphi(\mathsf{z})=\sum_{i=1}^p (1-\varphi_i \mathsf{z})^{i}$ and $\theta(\mathsf{z}) = \sum_{i=1}^q (1-\vartheta_i \mathsf{z})^i $ lie outside the unit circle and let the polynomials have no common roots. Denote the covariances $\gamma_{|t-j|} := cov(\epsilon_t,\epsilon_j)$ and let the initial value of the explosive process be $x_0 = 0$.
\end{assumption}

\begin{lemma} \label{lem:ARMA_znonzero}
    Assumption \ref{ass:arma} implies Assumption \ref{ass:z_0}.
\end{lemma}

\subsubsection{Convergence and Independence of $z$ and $y$} \label{sec:indep}
We begin establishing that the moment condition and stationarity are sufficient for marginal convergence in distribution of $y_T$.

\begin{lemma} \label{lem:y_dist_stat}
    Under Assumption \ref{ass:var} and stationary innovations, there exists a distribution $y$ with a finite second moment, such that $ y_T \overset{D}{\rightarrow} y $.
\end{lemma}

\noindent However, stationarity does not necessarily imply joint convergence of $(y_T, z_T)$. This is demonstrated in the following counterexample:
\begin{example}
    Let $\Omega_1, \Omega_2$ be disjoint sets satisfying $\mathbb{P}(\Omega_1)=\mathbb{P}(\Omega_2)=0.5$ and $\Omega_1 \cup \Omega_2 = \Omega$. Assume the initial value is $x_0=0$. Define the innovation sequence $(\epsilon_t)_{t \in \mathbb{N}}$ via:
    \begin{equation*}
        \big(\epsilon_t(\omega)\big)_{t \in \mathbb{N}} = \begin{cases}
            \big( (-1)^t \big)_{t \in \mathbb{N}} \hspace{11pt} \text{ for } \omega \in \Omega_1 \\
            \big( (-1)^{t+1} \big)_{t \in \mathbb{N}} \text{ for } \omega \in \Omega_2.
        \end{cases}
    \end{equation*}
    That is $(\epsilon_t)_{t \in \mathbb{N}}$ oscillates deterministically between $ -1$ and $1$, with the only source of randomness being whether even or odd indexes take on positive values. Note this process is stationary. \\ \leavevmode
    \\
    \noindent However, in this case $(y_T,z_T)$ is divergent, as convergence along \emph{even} $T$ and \emph{odd} $T$ lead to different CDFs:
    \begin{align*}
        \text{ T odd:} \quad &\, {(1+\beta)}{(y_T, z_T)} \overset{D}{\rightarrow} \begin{cases}
            (-1, -1) \text{ w.p. } 0.5 \\
            (\hspace{10pt} 1, \hspace{10pt}  1)  \text{ w.p. } 0.5
        \end{cases} \\
        \text{ T even:}\quad &  \, {(1+\beta)}(y_T, z_T) \overset{D}{\rightarrow} \begin{cases}
            (-1, \hspace{10pt} 1) \text{ w.p. } 0.5 \\
            ( \hspace{10pt} 1, -1)  \text{ w.p. } 0.5
        \end{cases}
    \end{align*}
    That is along the subsequence of odd integers, $z_T, y_T$ 
    converge to the same limit, whilst along the subsequence of even integers $z_T, y_T$ converge to limits of opposite signs. 
\end{example}

\noindent We now show that $\alpha$-mixing is sufficient for joint convergence of $(y_T, z_T)$ as well as independence between the limiting objects. We recall the definition of $\alpha$-mixing 
\begin{assumption}[$\alpha$-mixing]  \label{ass:alpha_mixing}
    Define the $\alpha$-mixing coefficient between two $\sigma$-fields $\mathcal{A}, \mathcal{B}$ as 
    \[
    \alpha(\mathcal{A}, \mathcal{B}) := \sup_{A \in \mathcal{A}, B \in\mathcal{B}} \left| \mathbb{P}(A \cap B) - \mathbb{P}(A) \mathbb{P}(B) \right|.\]
    The $\alpha$-mixing coefficient of the process $(\epsilon_t)$ is defined via 
    \[ \alpha_\epsilon(h) := \sup_u \alpha \big( \sigma(x_0, \epsilon_t: t \leq u),  \sigma(\epsilon_t: t \geq u+h) \big).\]
    The innovations satisfy $\alpha_\epsilon(h) \rightarrow 0$ as $h \rightarrow \infty$.
\end{assumption}

\begin{remark} \label{rem:mixing}
   By \cite[Theorem 6, pp. 99]{Dou1994}, ARMA innovations as in Assumption \ref{ass:arma} are $\beta$-mixing, which implies $\alpha$-mixing.
\end{remark}

\noindent We now generalize \cite[Theorem 2.3]{And1959} by lessening his independence requirement to $\alpha$-mixing in the following theorem.
\begin{theorem} \label{lem:independence}
    Under Assumptions \ref{ass:var} and \ref{ass:alpha_mixing}, and if $y_T \overset{D}{\rightarrow} y$, then $(y_T, z_T) \overset{D}{\rightarrow} (y,z)$, where $y$ and $z$ are independent and both have a finite second moment.
\end{theorem}
\begin{corollary}
    Suppose the innovations are stationary and satisfy Assumptions \ref{ass:var}, \ref{ass:alpha_mixing}. Then $(y_T, z_T) \overset{D}{\rightarrow} (y,z)$, where $y$ and $z$ are independent and both have a finite second moment.
\end{corollary}
\begin{remark} The corollary is established by noting that \emph{stationarity} of $(\epsilon_t)_{t \in \mathbb{N}}$ is a sufficient condition for $y_T \overset{D}{\rightarrow} y$ by Lemma \ref{lem:y_dist_stat}. Theorem \ref{lem:independence} assumes $y_T \overset{D}{\rightarrow} y$ as this also encompasses situations where $y_T$ converges marginally even when the innovations are nonstationary.
\end{remark}

\section{Implications} \label{sec:imp}
In this section, we will apply some of the previous general results to demonstrate their use cases. For inference we need to know the distribution of $y$. In general, nothing ensures that this variable is Gaussian or has a standard CDF. We show however, that if the innovations follow a Gaussian ARMA-process, the exact distribution of $y$ can be derived. This shows that feasible inference in the explosive setting goes much beyond the standard iid Gaussian setting originally provided in \cite{And1959}. In another example, we show that the presented AR(1) results can be extended to an AR(2) with intercept, demonstrating extendability of the AR(1) results to larger models.

\subsection{Results for Gaussian ARMA Innovations} \label{sec:arma}
Suppose the innovations follow a stationary Gaussian ARMA(p,q) as in Assumption \ref{ass:arma} and recall that $(\gamma_h)_{h \in \mathbb{N}}$ are the covarainces of the process. One may then define $\Gamma = \gamma_0 + 2 \sum_{h=1}^\infty \gamma_h \beta^h$. This leads to the following behavior of $(y_T, z_T)$.
\begin{lemma} \label{lem:arma}
    Under Assumption \ref{ass:arma}, $(y_T, z_T) \overset{D}{\rightarrow} N(0, \sigma^2 I_2)$, with $\sigma^2 = \frac{\Gamma}{1-\beta^2}.$
\end{lemma}
\noindent Stable Gaussian ARMAs as in Assumption \ref{ass:arma} satisfy Assumptions \ref{ass:var}-\ref{ass:indep} by the following arguments: Stable ARMA innovations have a finite second moment, so Assumption \ref{ass:var} is satisfied. Assumption  \ref{ass:z_0} is satisfied for ARMA innovations by Lemma \ref{lem:ARMA_znonzero}. Finally, Lemma \ref{lem:arma} implies Assumptions \ref{ass:y_exists}-\ref{ass:indep}. Moreover, as Gaussian ARMA innovations without an intercept have a symmetric unconditional distribution around zero, the symmetry condition for Corollary \ref{cor:symm} is also satisfied. Therefore, combining Lemma \ref{lem:arma} with Corollary \ref{cor:symm}, requiring Assumptions \ref{ass:var}-\ref{ass:indep} and the previous symmetry condition, yields the following proposition.
\begin{proposition} \label{thm:arma}
    Under Assumption \ref{ass:arma}, jointly,
    \begin{align*}
        \frac{\rho^T}{\rho^2-1} (\hat{\rho} - \rho) &\overset{D}{\rightarrow} Cauchy \\
        \frac{A_T}{\sqrt{B_T}} &\overset{D}{\rightarrow} N \big(0, \Gamma \big).
    \end{align*}
\end{proposition}
\begin{remark}
    \cite{Whi1958} proves that under iid Gaussian errors, the scaled OLS estimator has a limiting Cauchy distribution. Proposition \ref{thm:arma} demonstrates the same limit holds even for potentially autocorrelated stable Gaussian ARMA(p,q) innovations.
\end{remark}

\begin{remark}
    \cite{And1959} demonstrates in the Gaussian iid case, $A_T / \sqrt{B_T}$ has a limiting Gaussian distribution with variance $\gamma_0$. However, in the ARMA(p,q) case, the dependence in the innovations distorts the variance of the t-statistic via the term $2 \sum_{h=1}^\infty \gamma_h \beta^h$.
\end{remark}

 \noindent As the estimator $\hat{\rho}$ is rate $\rho^T$ consistent, the residuals of the simple regression can be used to consistently recover the autocovariances of the innovations $(\gamma_h)_{h \in \mathbb{N}}$, which may be used to achieve a feasible test statistic for an explosive AR(1) with ARMA innovations. We do so by first establishing some uniform consistency results, beginning with the estimator $\hat{\beta}^h := \hat{\rho}^{-h}$, which also holds beyond the ARMA setting. 

\begin{lemma} \label{lem:uniformbeta}
    Under Assumptions \ref{ass:var}-\ref{ass:z_0}, one has $\sup_{0<h<T} |\hat{\beta}^h-\beta^h| \overset{a.s.}{\rightarrow}0$. Furthermore, one has $\sum_{h=1}^T |\hat{\beta}^h - \beta^h| \overset{a.s.}{\rightarrow} 0$.
\end{lemma}
\noindent Next, let $L_T$ denote a sequence satisfying $L_T \rightarrow \infty$ and $L_T T^{-1} \rightarrow 0$. Define the sample autocovariances for $h \in \mathbb{N}_0$ as
\begin{equation*}
    \hat{\gamma}_h := (T-h)^{-1} \sum_{t=h+1}^T (x_t - \hat{\rho} x_{t-1})(x_{t-h}-\hat{\rho}x_{t-h-1}).
\end{equation*}
One can then show
\begin{lemma}\label{lem:ACFs}
    Let the innovations satisfy Assumption \ref{ass:arma}. Then 
    \begin{equation*} \textstyle
        \sup_{0 \leq h\leq L_T} |\hat{\gamma}_h - \gamma_h| \overset{\mathbb{P}}{\rightarrow} 0.
    \end{equation*}
\end{lemma}
\noindent This enables us to define a feasible variance estimator of $\Gamma$ as
\begin{equation*}
    \hat{\Gamma} := \hat{\gamma}_0 + 2 \sum_{h=1}^{L_T} \hat{\gamma}_h \hat{\beta}^h.
\end{equation*}
Which can be combined with Lemmas \ref{lem:uniformbeta}, \ref{lem:ACFs} to establish consistency of $\hat{\Gamma}$, enabling inference on $\rho$ in the ARMA setting. 
\begin{theorem} \label{thm:feasiblearma}
    Under Assumption \ref{ass:arma}, then $\hat{\Gamma} \overset{\mathbb{P}}{\rightarrow} \Gamma$ and $ \hat{\Gamma}^{-1} \frac{A_T}{\sqrt{B_T}} \overset{D}{\rightarrow} N(0,1).$
\end{theorem}

\subsection{Consistency of AR(2)} \label{sec:ar2}
Based on \cite{LaiWei1983a} and \cite{LaiWei1985}, \cite{Nie2005} provides general consistency results for VAR(p) models with deterministic terms. Those results require the innovations to satisfy the \emph{Marcinkiewicz-Zygmund} moment condition
\begin{equation} \label{eq:mz}
    \exists \delta>0 : sup_t \mathbb{E}[\epsilon_t^{2+\delta}|\mathcal{F}_{t-1}] < \infty \quad a.s.
\end{equation}
where $(\mathcal{F}_{t})$ is an increasing sequence of $\sigma$-fields. This is considered quite strong for most time-series contexts, as it for instance may exclude latent-state processes and ARCH. However, in Section \ref{sec:econometric_theory}, we demonstrated that Assumptions \ref{ass:var}-\ref{ass:z_0} were sufficient for consistency of $\hat{\rho}$. Moreover, we provided sufficient conditions for latent-state innovations to satisfy Assumption \ref{ass:z_0} in Lemma \ref{lem:SV}, providing an alternative set of sufficient conditions for consistency that may be used when the Marcinkiewicz-Zygmund condition in (\ref{eq:mz}) is violated or difficult to verify. This alternative set of sufficient conditions for consistency in the AR(1) can be carried over to general autoregressions by re-applying the same diagonalization framework as \cite{Nie2005}. This is demonstrated in the following example, where we show how the AR(1) results carry over to an AR(2) with intercept.
\begin{example}{\label{ex:ar2}}
Consider the explosive AR(2) with intercept $\mu \in \mathbb{R}$ of the form 
\begin{equation*}
    (1-\rho L) (1-\alpha L) x_t = \mu + \epsilon_t, \quad 0<\alpha<1<\rho.
\end{equation*}
By stacking terms into a vector $\mathbf{x}_t :=(x_t, x_{t-1}, 1)$ and defining
\begin{equation*}
    \Pi := \begin{pmatrix}
        \alpha+\rho & -\alpha \rho & \mu \\
        1 & 0 & 0 \\ 0 & 0 & 1
    \end{pmatrix}, \quad M := \begin{pmatrix}
        1 & -\rho & \tfrac{\mu}{\alpha-1} \\ 0 & 0 & 1 \\ 1 &-\alpha & \tfrac{\mu}{\rho-1} 
    \end{pmatrix} , \quad \Lambda:= \begin{pmatrix}
        \alpha & 0 & 0 \\ 0 & 1 & 0 \\ 0 & 0 & \rho
    \end{pmatrix}
\end{equation*}
one may write the model in companion form as $\mathbf{x}_t = \Pi \mathbf{x}_{t-1} + e_t$, where $e_t = (\epsilon_t,0,0)$. As $M \Pi = \Lambda M$, the companion model  can pre-multiplied by $M$, delivering a diagonalized process $(u_t, 1, w_t) := M\mathbf{x}_t$ of the form
\begin{equation*}
    \begin{pmatrix}
        u_t \\ 1 \\ w_t
    \end{pmatrix} =  \begin{pmatrix}
        \alpha & 0 & 0 \\ 0 & 1 & 0 \\ 0 & 0 & \rho
    \end{pmatrix} \begin{pmatrix}
        u_{t-1} \\ 1 \\ w_{t-1}
    \end{pmatrix} + Me_t.
\end{equation*}
As $M$ is invertible, $x_t$ can be expressed as a linear combination of a stable AR(1), an explosive AR(1) and a deterministic term. \\ \leavevmode \\
This transformation is especially useful if we wish to analyse multivariate least-squares estimators. Define $\hat{\theta}$ as the least-squares estimator arising from a regression of $x_t$ on $x_{t-1}, x_{t-2}$ and an intercept. Defining $Y:=(\mathbf{x}_{T-1}'M', \mathbf{x}_{T-2}'M', ..., \mathbf{x}_{0}'M')$ and $\epsilon = (\epsilon_T, ..., \epsilon_1)$, one may write $\hat{\theta}$, centered around its true value $\theta = (\alpha+\rho, -\alpha \rho, \mu)$ as 
\begin{equation*}
    (\hat{\theta} - \theta) = M' (Y'Y)^{-1}Y'\epsilon.
\end{equation*}
Defining
\begin{equation*}
    D_T:= \begin{pmatrix}
        \{\sum_{t=1}^T{u_{t-1}^2}\}^{1/2} & & \\
        & \sqrt{T} & \\
        && \{ \sum_{t=1}^T w_{t-1}^2\}^{1/2}    \end{pmatrix}
\end{equation*}
and
\begin{equation*}
    H_T := \begin{pmatrix}
        0 & \hat{\phi}(u_{t-1},1) & \hat{\phi}(u_{t-1},w_{t-1}) \\
        \hat{\phi}(1, u_{t-1}) & 0 & \hat{\phi}(1,w_{t-1}) \\
        \hat{\phi}(w_{t-1},u_{t-1}) & \hat{\phi}(w_{t-1},1) & 0
    \end{pmatrix}
\end{equation*}
and finally
\begin{equation*}
    \hat{\phi}(a_t,b_t):=\frac{\sum_{t=1}^T a_t b_t}{\sqrt{\sum_{t=1}^T a_t^2} \sqrt{\sum_{t=1}^T b_t^2}},
\end{equation*}
one may decompose $(\hat{\theta}-\theta)$ as
\begin{equation*}
     (\hat{\theta}-\theta) = M' (D_T [I_3 + H_T] D_T)^{-1} \sum_{t=1}^T
     \begin{pmatrix}
        u_{t-1} \epsilon_t \\
        \epsilon_t \\
        w_{t-1} \epsilon_t
    \end{pmatrix}.
\end{equation*}
This form can for instance be used to establish consistency of $\hat{\theta}$. As the focus of this paper is the explosive component, we assume the innovations and stable AR(1) component $u_t$ satisfy the follow high-level behavior:
\begin{assumption} \label{ass:stableAR}
    Let $\sigma_u^2, M^2 \in \mathbb{R}_+$. Moreover, assume $T^{-1} \sum_{t=1}^T \epsilon_t u_{t-1} \overset{\mathbb{P}}{\rightarrow} 0$, $T^{-1} \sum_{t=1}^T \epsilon_t \overset{\mathbb{P}}{\rightarrow} 0$, $T^{-1} \sum_{t=1}^T u_{t-1}^2 \overset{\mathbb{P}}{\rightarrow} \mathbb{E}[u_{t-1}^2] =\sigma_u^2$, $T^{-1} \sum_{t=1}^T u_{t-1} \overset{\mathbb{P}}{\rightarrow} 0$ and assume that $T^{-1} \sum_{t=1}^T \epsilon_t^2 \overset{\mathbb{P}}{\rightarrow} M^2$.
\end{assumption}
\noindent With high-level assumptions in place ensuring stable components are well-behaved, we may use the explosive convergence results provided in this paper to demonstrate that terms containing the explosive component in $(\hat{\theta}-\theta)$ 
vanish, and $\hat{\theta}$ is therefore consistent.  
\begin{proposition} \label{prop:consistency}
     Let the innovations satisfy Assumptions \ref{ass:var}-\ref{ass:z_0} and \ref{ass:stableAR}. Then $\hat{\theta} \overset{\mathbb{P}}{\rightarrow} \theta$.
\end{proposition}
    
\end{example}

\section{Conclusion}
We generalized Anderon's first result, showing that the assumptions of \emph{uncorrelated} and \emph{independent} innovations were not necessary for the numerator and denominator to converge to forward- and backward averages of the innovations. We then demonstrated how the independence assumption of the innovations made in \cite[Theorem 2.3]{And1959} may be relaxed to $\alpha$-mixing. We provided the exact limiting distributions of the forward- and backward averages in the ARMA setting, demonstrating that inference is feasible outside the $iid$ setting. Finally, we demonstrated how the AR(1) results, in particular the consistency results, may be extended to larger autoregressions via an example.

\newpage
\appendix
\section{Proofs}

\subsection{Preliminary Results}
Below, we collect some preliminary results that are used repeatedly to prove the main results. The proofs generally proceed by bounding moments and then using these bounds to conclude almost sure convergence. This subsection is structured to follow that workflow explicitly: Lemma \ref{lem:reduction_lemma} demonstrates how a useful moment condition implies almost sure convergence, Lemma \ref{lem:cauchy} proves that $z_T$ is a Cauchy sequence, Lemma \ref{lem:z_T_bounds} and Lemma \ref{lem:z_bounds} provide moment bounds for $y_T, z_T$. Generally, objects defined within proofs only retain their definitions within the proof.

\begin{lemma}[Reduction Lemma] \label{lem:reduction_lemma}
    Let $(S_T)$ be a sequence of random variables on $(\Omega, \mathcal{F}, \mathbb{P})$. The moment condition $ \sum_{T=1}^\infty \mathbb{E}|S_T| < \infty $ implies $S_T \overset{a.s.}{\rightarrow} 0$.
\end{lemma}
\noindent \textbf{Proof:}
By Markov's inequality
\begin{equation*}
    \forall \delta>0: \sum_{T=1}^\infty P(|S_T| \geq \delta) \leq \frac{1}{\delta} \sum_{T=1}^\infty E |S_T| < \infty.
\end{equation*}
By a Borel-Cantelli argument, see Theorem 2.1.1  \cite{Stout74}, this implies $ S_T \overset{a.s.}{\rightarrow} 0._{\qedsymbol{}}$

\begin{lemma} \label{lem:cauchy}
    Under Assumption \ref{ass:var}, $z_T$ is a Cauchy sequence in $\mathbb{R}$, almost surely.
\end{lemma}
\noindent \textbf{Proof:} 
Define $d_T := \sup_{m \geq T} | z_m - z_T |$. By the triangle inequality
\begin{equation*}
    \mathbb{E} \left[ d_T \right] =  \mathbb{E} \left[ \sup_{m \geq T} \bigg| \sum_{t=T}^{m-1} \beta^{t-1} \epsilon_t \bigg| \right] \leq \mathbb{E} \left[ \sup_{m \geq T}  \sum_{t=T}^{m-1} |\beta|^{t-1} |\epsilon_t| \right].
\end{equation*}
Use the monotone convergence theorem to swap the expectation and the supremum, then bound $\mathbb{E}|\epsilon_t| \leq M$ via Assumption \ref{ass:var}, delivering 
$\mathbb{E} \left[ d_n \right] \leq \beta^T {M}{(1-|\beta|)^{-1}}$. By Lemma \ref{lem:reduction_lemma}, the bound implies $d_T \overset{a.s.}{\rightarrow} 0._{\qedsymbol{}}$

\begin{lemma}[bounds on $z_T, y_T$] \label{lem:z_T_bounds}
Under Assumption \ref{ass:var} there exists a finite $K \in \mathbb{R}_+$, such that for all integers $1\leq s<T$\\
\begin{minipage}{0.42\textwidth}
  \begin{align}
     \mathbb{E}[z_T^2] & \leq K \label{eq:z_L2} \\[6pt]
     \mathbb{E}[y_T^2] & \leq K \label{eq:y_L2}
    \end{align}
\end{minipage}
\begin{minipage}{0.55\textwidth}
    \begin{align}
    \mathbb{E}\left[ (z_T - z_{T-s})^2 \right] & \leq K \beta^{2(T-s)} \label{eq:z_findif_sqaured} \\[6pt]
    \mathbb{E}\left[ (z_T + z_{T-s})^2 \right] & \leq K \label{eq:z_sum_L2} 
    \end{align}
\end{minipage}
\end{lemma}
\noindent \textbf{Proof:} \\
\noindent \emph{Proof of (\ref{eq:z_L2}):} Recall that  $z_T := \beta^{T-2} x_{T-1}$ by (\ref{eq:zt_def}), expanding the square yields
\begin{equation}
    \mathbb{E}[z_T]^2 = \sum_{i=1}^{T-1} \beta^{2(i-1)} \mathbb{E} \epsilon_i^2 + 2 \sum_{i=1}^{T-1} \sum_{j>i}^{T-1} \beta^{i+j-2} \mathbb{E} \epsilon_i \epsilon_j + 2 \rho \sum_{i=1}^{T-1} \beta^{i-1} \mathbb{E} [x_0 \epsilon_i] + \rho^2 \mathbb{E}x_0^2 \label{eq:expanded}.
\end{equation}
The final term is trivially $O(1)$, as $\mathbb{E}x_0^2 \leq M^2$ under Assumption \ref{ass:var}. The other three terms can be bounded using the uniform Cauchy-Schwarz bounds under Assumption \ref{ass:var} $\mathbb{E}\epsilon_t^2 \leq M^2$, $\mathbb{E}|\epsilon_t \epsilon_j| \leq M^2$ and $\mathbb{E}x_0^2 \leq M^2$, respectively, by extending the upper summation index from $T-1$ to $\infty$ delivering
\begin{equation*}
    \mathbb{E}[z_T^2] \leq M^2 \left[ \sum_{i=1}^\infty \beta^{2(i-1)} + 2 \sum_{i=1}^\infty \sum_{j>i}^\infty |\beta|^{i+j+2} + 2 \rho \sum_{i=1}^\infty |\beta|^{i-1} + \rho^2 \right] =: K_1.
\end{equation*}
This upper bound is uniform in $T$. The bound is finite, allowing us to define the $K_1$, as all the geometric sums are finite by $|\beta|<1$. \\
\leavevmode \\
\noindent \emph{Proof of (\ref{eq:y_L2})}: Recall that $y_T := \sum_{i=1}^T \beta^{T-i} \epsilon_i$ by (\ref{eq:y_def}). Hence, $y_T = \sum_{i=0}^{T-1} \beta^i \epsilon_{T-i}$. Expanding the square yields
\begin{equation*}
    \mathbb{E}y_T^2 = \mathbb{E} \Big[ \sum_{i=0}^{T-1} \beta^i \epsilon_{T-i} \Big]^2 = \sum_{i=0}^{T-1} \beta^{2i} \mathbb{E} \epsilon_{T-i}^2 + 2 \sum_{i=0}^{T-1} \sum_{j>i}^{T-1} \beta^{i+j} \mathbb{E} [\epsilon_{T-i} \epsilon_{T-j}].
\end{equation*}
Bounding this in a similar fashion to (\ref{eq:z_L2}) delivers $\mathbb{E}[y_T^2] \leq K_1$. \\
\leavevmode \\
\noindent \emph{Proof of (\ref{eq:z_findif_sqaured}):} Begin by changing the direction of summation, delivering
\begin{equation*}
    \mathbb{E}[ z_{T-s} - z_T]^2 = \mathbb{E} \Big[ \sum_{i=T-s}^{T-1} \beta^{i-1} \epsilon_i \Big]^2 = \beta^{2(T-s)} \mathbb{E} \Big[ \sum_{i=0}^{s-1} \beta^{i-1} \epsilon_{i+T-s} \Big]^2.
\end{equation*}
The expected squared sum can be bounded in a similar fashion to (\ref{eq:z_L2}), uniformly in $s,T$, delivering $\mathbb{E}[ z_{T-s} - z_T]^2 \leq \beta^{2(T-s)} K_1$. \\ 
\noindent \\
\noindent \emph{Proof of (\ref{eq:z_sum_L2}):} Note that for $a,b \in \mathbb{R}: (a+b)^2 \leq 2a^2 + 2b^2$. A combination of this with the previous bound (\ref{eq:z_L2}) yields
    \begin{align*}
         \mathbb{E}[ z_T + z_{T-s}]^2 \leq \hspace{7pt} & 2 \mathbb{E}z_{T}^2 + 2 \mathbb{E}z_{T-s}^2 \leq 4 \sup_{T \in \mathbb{N}} \mathbb{E}z_T^2 \leq 4K_1.
    \end{align*}       

\noindent To complete the lemma, define $K := 4 K_1 ._{\qedsymbol{}}$

\begin{lemma}[square integrability and bounds on $z$] \label{lem:z_bounds}
    Under Assumption \ref{ass:var} there exists a finite $K \in \mathbb{R}_+$, such that for all integers $1 \leq s \leq T$ \\
    \begin{minipage}{0.42\textwidth}
  \begin{align}
       \mathbb{E}[z^2] &\leq K \label{eq:zsq_bound} \\
      \mathbb{E} | z - z_T | & \leq K |\beta|^T  \label{eq:z_absl_diff}
    \end{align}
\end{minipage}
\begin{minipage}{0.55\textwidth}
    \begin{align}
     \mathbb{E} \left[ (z_T - z)^2 \right] & \leq K \beta^{2T} \label{eq:z_diff_sq} \\[6pt] 
     \mathbb{E} \left| z^2 -  z_T^2 \right| & \leq K |\beta|^{T}  \label{eq:z_squared_diff} 
    \end{align}
\end{minipage}
\end{lemma}
\noindent \textbf{Proof:}\\
\noindent \emph{Proof of (\ref{eq:zsq_bound}):} Recall $A = \{ \omega: \lim_T z_T \text{ exists in }  \mathbb{R} \} $ from the definition of $z$ in (\ref{eq:zdef}). By Lemma \ref{lem:cauchy}, requiring Assumption \ref{ass:var}, we have that $\mathbb{P}(A)=1$, and so that
\begin{align*}
    \int_\Omega z^2(\omega) d\mathbb{P} = \int_A \liminf_{T} z_T^2(\omega) d\mathbb{P} \leq \liminf_{T}  \int_A z_T^2(\omega) d\mathbb{P} \leq \sup_T \mathbb{E} z_T^2 \leq K_1.
\end{align*}
The first inequality uses Fatou's Lemma, see Theorem 16.3 of \cite{Bil79}. The final  final inequality uses (\ref{eq:z_L2}) of Lemma \ref{lem:z_bounds}, also requiring Assumption \ref{ass:var}. We give the constant $K_1$ a subscript to differentiate it from different bounds derived later on. \\

\noindent \emph{Proof of (\ref{eq:z_absl_diff}):} By Lemma \ref{lem:cauchy}, requiring Assumption \ref{ass:var}, $z \overset{a.s.}{=} \sum_{i=1}^\infty \beta^{i-1} \epsilon_i$. A combination of the almost sure identity with the triangle inequality yields
    \begin{equation*}
    \mathbb{E}| z - z_T |  = \mathbb{E} \left| \sum_{i=T}^\infty \beta^{i-1} \epsilon_i \right| \leq \mathbb{E} \, \underset{n \rightarrow \infty}{lim} \sum_{i=T}^n |\beta|^{i-1} |\epsilon_i|.
    \end{equation*}
    Note that $f_n :=  \sum_{i=T}^n |\beta|^{i-1} |\epsilon_i|$ is monotone in $n$. Apply the monotone convergence theorem to swap limits and integration, then apply the bound $\mathbb{E}|\epsilon_i| \leq M$. This yields
    \begin{equation*}
        \mathbb{E} \, \underset{n \rightarrow \infty}{lim} \sum_{i=T}^n |\beta|^{i-1} |\epsilon_i| = \sum_{i=T}^\infty |\beta|^{i-1} \mathbb{E} | \epsilon_i | \leq M \sum_{i=T}^\infty |\beta|^{i-1} = \frac{M |\beta|^{T-1}}{1-|\beta|} =: K_2 |\beta|^T.
    \end{equation*}
    \noindent \emph{Proof of (\ref{eq:z_diff_sq}):} Following a similar procedure, one may obtain 
    \begin{equation*} 
        \mathbb{E}[(z - z_T)^2] = \mathbb{E} \Big[ \big( \sum_{i=T}^\infty \beta^{i-1} \epsilon_i \big)^2 \Big] \leq \mathbb{E} \Big[\sum_{i=T}^\infty \beta^{2(i-1)} \epsilon_i^2 + 2 \sum_{i=T}^\infty \sum_{j>i}^{\infty} |\beta|^{i+j-2} |\epsilon_i \epsilon_j| \Big].
    \end{equation*}
    Note that $f_n := \sum_{i=T}^n \beta^{2(i-1)} \epsilon_i^2$ and $h_n := \sum_{i=T}^n \sum_{j>i}^{n} |\beta|^{i+j-2} |\epsilon_i \epsilon_j|$ are increasing monotonically in $n$. Therefore, we apply the Monotone Convergence Theorem and moment bounds, delivering
    \begin{equation*}
       \mathbb{E}(z - z_T)^2 \leq  M^2 \left(\sum_{i=T}^\infty \beta^{2(i-1)} + 2 \sum_{i=T}^\infty \sum_{j>i}^{\infty} |\beta|^{i+j-2} \right) = \frac{M^2 \beta^{2(T-1)}}{(1-|\beta|)^2} =: K_3 \beta^{2T}.
    \end{equation*}
    \noindent \emph{Proof of (\ref{eq:z_squared_diff}):} By adding and subtracting $z_T$ conveniently, then applying the triangle inequality, one may obtain
\begin{equation} \label{eq:twoterms}
    \mathbb{E}|(z^2 - z_T^2)| = \mathbb{E}|(z - z_T + z_T)^2 - z_T^2| \leq \mathbb{E}[(z - z_T)^2] + 2 \mathbb{E}|(z-z_T)z_T|.
\end{equation}
The first term satisfies $\mathbb{E}(z - z_T)^2 \leq K_3 \beta^{2T}$ by (\ref{eq:z_diff_sq}). Applying the Cauchy-Schwarz inequality on the second term, noting $\mathbb{E}(z - z_T)^2 \leq K_3 \beta^{2T}$ by (\ref{eq:z_diff_sq}) as well as $\mathbb{E} z_T^2 \leq K_1$ by (\ref{eq:z_L2}) of Lemma \ref{lem:z_bounds}, requiring Assumption \ref{ass:var}, delivers
\begin{equation*}
   \mathbb{E}|(z-z_T)z_T| \leq \left\{ \mathbb{E}[(z-z_T)^2] \mathbb{E} z_T^2 \right\}^{1/2} \leq \{K_1 K_3 \beta^{2T} \}^{1/2} = \sqrt{K_1 K_3} |\beta|^{T}.
\end{equation*}
With both terms in (\ref{eq:twoterms}) bounded, as $\beta^{2T} \leq |\beta|^T$ for all integers $T$, we have the bound $ \mathbb{E}(z - z_T)^2  \leq  K_4 |\beta|^{T}$. \\ 
\leavevmode
\\
Once again, to conclude the lemma, select $K := \max \{K_i\}_{i=1}^4._{\qedsymbol{}}$

\begin{lemma}\label{lem:finite_l2}
    Under Assumption \ref{ass:var}, if there exists a distribution $y$ such that $y_T \overset{D}{\rightarrow}y$, then $y$ has a finite second moment.
\end{lemma}
\noindent \textbf{Proof:} Recall that by (\ref{eq:y_L2}) of Lemma \ref{lem:z_T_bounds}, there exists a finite $K$, such that $\sup_T\mathbb{E}[y_T^2]\leq K$. Let $\mu$ be the limiting law of $y_T$. It then follows by the Portmanteau Theorem that $\int x^2 d\mu (x) \leq \liminf_{T} \mathbb{E}[y_T^2] \leq \sup_{T} \mathbb{E}[y_T^2] \leq K._{\qedsymbol{}}$

\subsection{Proofs of Section \ref{sec:found_results}} \label{sec:app2}
This segment contains proofs of Section \ref{sec:econometric_theory}. They are presented sequentially, with auxiliary lemmas introduced and proven as needed.

\subsubsection{Proof of Lemma \ref{lem:measurability}}
Recall the $A = \{ \omega: \lim_T z_T \text{ exists in }\mathbb{R} \}$ from the definition of $z$ in (\ref{eq:zdef}). In Lemma \ref{lem:cauchy}, requiring Assumption \ref{ass:var}, we demonstrated that $z_T$ is a \emph{Cauchy sequence} in $\mathbb{R}$ with probability 1, and so that $\mathbb{P}(A)=1$. By the definition of $z$, it therefore follows that for all $\omega \in A$, $z_T(\omega) \rightarrow z(\omega)$ as $T \rightarrow \infty$. This implies $z_T \overset{a.s.}{\rightarrow} z$. Square integrability follows directly from the bound (\ref{eq:zsq_bound}) of Lemma \ref{lem:z_bounds}, again under Assumption \ref{ass:var}.

\subsubsection{Proof of Lemma \ref{lem:B_T}} \label{AP:B_T}
The convergence of the denominator, stated in (\ref{eq:den_conv}) follows from Theorem 2.1 of \cite{And1959}, whilst convergence of the numerator in (\ref{eq:num_conv}) follows  from Theorem 2.2 of \cite{And1959}. The contribution of this paper lies in the new moment bounds provided in Lemma \ref{lem:z_T_bounds} and Lemma \ref{lem:z_bounds}, which allows one to retrace the steps in \cite{And1959} without requiring uncorrelatedness or centering around zero. For completeness, those proofs are provided below.\\
\leavevmode \\
\noindent \emph{Proof of $\beta^{2(T-2)} B_T \overset{a.s.}{\rightarrow} z^2/{1-\beta^2}$:} Adding and subtracting $ z_T^2(1-\beta^{2T})/(1-\beta^2)$ yields
\begin{equation*}
    S_T^{(B)} := \beta^{2(T-2)} B_T - \tfrac{1}{1-\beta^2} z^2  =  \beta^{2(T-2)} B_T - \tfrac{1-\beta^{2T}}{1-\beta^2} z_T^2 + \tfrac{1-\beta^{2T}}{1-\beta^2} z_T^2 - \tfrac{1}{1-\beta^2} z^2 .
\end{equation*}
Recall that $B_T = \sum_{t=1}^T x_{t-1}^2$ and $z_T = \beta^{T-2} x_{T-1}$, so $\beta^{2(T-2)}B_T =  \sum_{s=0}^{T-1} \beta^{2s} z_{T-s}^2$. Also note the identity $z_T({1-\beta^{2T}})/({1-\beta^2}) = \sum_{s=0}^{T-1} \beta^{2s} z_{T}^2$. This allows the previous expression to be re-written as
\begin{equation*}
    S_T^{(B)} = \textstyle\sum\nolimits_{s=1}^{T-1} \beta^{2s} ( z_{T-s}^2 -z_T^2) + \tfrac{1-\beta^{2T}}{1-\beta^2} z_T^2 - \tfrac{1}{1-\beta^2} z^2.
\end{equation*}
This identity combined with the triangle inequality yields yields the bound 
\begin{equation*} \textstyle
     \mathbb{E} | S_T^{(B)} | \leq \mathbb{E} \left| \sum_{s=1}^{T-1} \beta^{2s} (z_{T-s}^2 - z_T^2) \right| + \frac{1}{1-\beta^2} \mathbb{E}|z_T^2 - z^2| +  \frac{\beta^{2T}}{1-\beta^2} {\mathbb{E}| z_T^2 |}.
\end{equation*}
By the third binomial equation, $(z_{T-s}^2 - z_T^2) = (z_{T-s} - z_T)(z_{T-s} + z)$. A combination of this with Cauchy-Schwarz inequality yields
\begin{equation*}
    \textstyle \mathbb{E} | S_T^{(B)} | \leq \sum_{s=1}^{T-1} \beta^{2s} \Big\{ \mathbb{E}[(z_{T-s} - z_T)^2] \mathbb{E}[(z_{T-s} + z_T)^2] \Big\}^{1/2} + \tfrac{1}{1-\beta^2} \mathbb{E}|z_T^2 - z^2| +  \tfrac{\beta^{2T}}{1-\beta^2} {\mathbb{E}| z_T^2 |}.
\end{equation*}
We now substitute in the new moment bounds. Under Assumption \ref{ass:var}, by Lemma \ref{lem:z_T_bounds} and Lemma \ref{lem:z_bounds}, there exists a constant $K$, such that $\mathbb{E}|z_T^2 - z^2| \leq K |\beta|^T$, $\mathbb{E}z_T^2 \leq K$, $\mathbb{E}[(z_{T-s} - z_T)^2] \leq K \beta^{2(T-s)}$ and $\mathbb{E}[(z_{T-s} + z_T)^2] \leq K$. It follows that 
\begin{equation*}
    \textstyle \mathbb{E} | S_T^{(B)} | \leq K \left(  \sum_{s=1}^{T-1} \beta^{2s} \left\{ \beta^{2(T-s)} \right\}^{1/2} + \frac{|\beta|^{T}}{1-\beta^2} + \frac{\beta^{2T}}{1-\beta^2} \right) .
\end{equation*}
Bounding the geometric series then yields the bound
\begin{equation}  \label{eq:interim_bound}
    \textstyle \mathbb{E} | S_T^{(B)} | \leq K  \left( \frac{|\beta|^{T}}{1-|\beta|} + \frac{\beta^{2T}+|\beta|^{T}}{1-\beta^2} \right).
\end{equation}
As this bound implies $\sum_{T=1}^\infty \textstyle \mathbb{E} | S_T^{(B)} | < \infty$, by Lemma \ref{lem:reduction_lemma}, $S_T^{(B)} \overset{a.s.}{\rightarrow}0$. \\
\leavevmode \\
\noindent \emph{Proof of $ (\beta^{T-2} A_T - y_T z) \overset{a.s.}{\rightarrow} 0$}: \\
Recall that $A_T = \sum_{t=1}^T \epsilon_t x_{t-1}$, $z_{t} = \beta^{t-2} x_{t-1}$ and $y_T = \sum_{t=1}^T \beta^{T-t} \epsilon_t$. We therefore have that
\begin{equation*}
 S_T^{(A)} := \textstyle \beta^{T-2} A_T - y_T z = \sum_{t=1}^T \beta^{T-t} \epsilon_{t} (z_{t} - z).
\end{equation*}
A bound may be achieved via the triangle and Cauchy-Schwarz inequalities, yielding
\begin{equation*} \textstyle
 \mathbb{E} |S_T^{(A)}|  \leq \sum_{t=1}^T |\beta|^{T-t} \mathbb{E} | \epsilon_{t} (z_t - z) | \leq \sum_{t=1}^T |\beta|^{T-t} \big\{ (\mathbb{E} \epsilon_{t}^2) \mathbb{E}[(z_t - z)^2] \big\}^{1/2}.
\end{equation*}
Under Assumption \ref{ass:var}, $\mathbb{E} \epsilon_t^2 \leq M^2$ and by Lemma \ref{lem:z_bounds}, also requiring Assumption \ref{ass:var}, there exists a uniform bound  $\mathbb{E}(z_{t} - z)^2 \leq K\beta^{2t}$. These can be used to derive the bound
\begin{equation} \label{eq:noice_bound} \textstyle
 \mathbb{E}|S_T^{(A)}|  \leq \sum_{t=1}^T |\beta|^{T-t} \{M^2 K \beta^{2t}\}^{1/2} = T \beta^{T} \{M^2 {K} \}^{1/2}.
\end{equation}
By pp. 71-72 of \cite{pri03}, $\sum_{T=0}^\infty T |\beta|^T$ is a convergent power series. Therefore, $\sum_{T=1}^\infty \mathbb{E}|S_T^{(A)}| < \infty$ and by Lemma \ref{lem:reduction_lemma}, $S_T^{(A)} \overset{a.s.}{\rightarrow} 0._{\qedsymbol{}}$

\begin{lemma} \label{lem:halfway}
    Under Assumption \ref{ass:var}, denoting $\mathsf{x}_{2,T} := (1-\beta^2)\beta^{2(T-2)}B_T$ we have:
    \begin{equation*}
        |y_T| \left|\sqrt{\mathsf{x}_{2,T}}-|z| \right| \overset{a.s.}{\rightarrow} 0
    \end{equation*}
\end{lemma}
\noindent \textbf{Proof.} We begin by remarking that for any $a,b \geq 0$, one has
\begin{equation*}
    \Big(\sqrt{a} - \sqrt{b} \Big)^2 = a + b - 2 \sqrt{ab} \leq a + b - 2 \min\{ a,b\} = |a-b|.
\end{equation*}
Also recall that by the bound (\ref{eq:interim_bound}), derived under Assumption \ref{ass:var}, there exists a finite $C$, such that $\mathbb{E} \left| \mathsf{x}_{2,T} - z^2 \right| \leq C |\beta|^T$. A combinination the previous inequality with this $L^1$ bound yields
\begin{equation} \label{eq:noice_bound_again}
    \mathbb{E} \left[(\sqrt{\mathsf{x}_{2,T}} - |z| )^2 \right] \leq \mathbb{E} \left| \mathsf{x}_{2,T} - z^2 \right| \leq C |\beta|^T.
\end{equation}
By an application of the Cauchy-Schwarz inequality, then applying the bounds (\ref{eq:noice_bound_again}) and $\mathbb{E} y_T^2 \leq K$ provided in (\ref{eq:y_L2}) of Lemma \ref{lem:z_bounds}, both of which require Assumption \ref{ass:var}, one may obtain the bound
\begin{equation*}
    \mathbb{E}\left[ |y_T| \left|\sqrt{\mathsf{x}_{2,T}}-|z| \right| \right] \leq \left\{ \mathbb{E}y_T^2 \mathbb{E} \left[  (\sqrt{\mathsf{x}_{2,T}}-|z|)^2 \right] \right\}^{1/2}  \leq \left\{C K \right\}^{1/2} \beta^{\frac{T}{2}}.
\end{equation*}
This bound is summable. By Lemma \ref{lem:reduction_lemma}, conclude $|y_T| |\sqrt{\mathsf{x}_{2,T}}-|z| | \overset{a.s.}{\rightarrow} 0._{\qedsymbol{}}$

\begin{lemma} \label{lem:num_conv_zero}
    Let $f_T$ be a sequence converging to $0$ as $T \rightarrow \infty$. Then, under Assumptions \ref{ass:var} and \ref{ass:z_0},
    \begin{equation*}
        f_T \beta^{T-2}A_T \overset{\mathbb{P}}{\rightarrow} 0
    \end{equation*}
    Additionally, if $\sum_{T=1}^\infty |f_T| < \infty$, then the convergence holds almost surely.
\end{lemma}
\noindent \textbf{Proof.} Recalling the identity $ \beta^{T-2}A_T = \sum_{t=1}^{T} \beta^{T-t} \epsilon_{t} z_t$, see that
\begin{equation*} \textstyle
    S_T^{(C)} := f_T \beta^{T-2} A_T = f_T \sum_{t=1}^{T} \beta^{T-t} \epsilon_{t} z_t.
\end{equation*}
Bound this by applying the triangle and Cauchy-Schwarz inequalities, yielding
\begin{equation*}
    \mathbb{E} |S_T^{(C)}| \leq |{f_T}| \sum_{t=1}^{T}  |  \beta|^{T-t}  \mathbb{E} \left| \epsilon_{t}z_{t}  \right| \leq |{f_T}| \sum_{t=1}^{T} |\beta| ^{T-t} \left( \mathbb{E} \epsilon_{t}^2 \mathbb{E} z_{t}^2 \right)^{1/2}.
\end{equation*}
Under Assumption \ref{ass:var}, $\mathbb{E}\epsilon_t^2 \leq M^2$ and the bound $\mathbb{E} z_T^2 \leq K$ of Lemma \ref{lem:z_bounds} applies. Applying the bounds delivers
\begin{equation*}
    \mathbb{E} |S_T^{(C)}| \leq |{f_T}| \left\{M^2K\right\}^{1/2} \sum_{t=1}^{T} |\beta|^{T-t} \leq |f_T| \frac{\left\{M^2K\right\}^{1/2}}{1-|\beta|}.
\end{equation*}
As $f_T$ converges to $0$, $ S_T^{(C)} \overset{L_1}{\rightarrow} 0$ follows, implying $S_T^{(C)} \overset{\mathbb{P}}{\rightarrow} 0$.\\
If moreover, $f_T$ satisfies $\sum_{t=T}^\infty |f_T| < \infty$, the bound becomes absolutely summable. In that case, Lemma \ref{lem:reduction_lemma} applies, yielding $S_T^{(C)} \overset{a.s.}{\rightarrow} 0._{\qedsymbol{}} $

\subsubsection{Proof of Theorem \ref{thm:consistency}} 
Stack the strong convergence results of Lemmas \ref{lem:measurability}, \ref{lem:B_T} and \ref{lem:halfway}, requiring Assumption \ref{ass:var}, in a vector $ \mathbf{x}_T = (\mathsf{x}_{1,T}, \mathsf{x}_{2,T}, \mathsf{x}_{3,T}, \mathsf{x}_{4,T})$, converging to $\mathbf{x} = (\mathsf{x}_1, \mathsf{x}_2, \mathsf{x}_3, \mathsf{x}_{4})$, delivering
\begin{equation*}
    \mathbf{x}_T := \begin{pmatrix}
        z_T \\
        (1-\beta^2)\beta^{2(T-2)} B_T \\ 
          \beta^{T-2} A_T - y_T z  \\
         |y_T| \left|\sqrt{\mathsf{x}_{2,T}}-|z| \right|
    \end{pmatrix}
    \overset{a.s.}{\rightarrow} \begin{pmatrix}
        z \\
        z^2 \\ 
        0 \\
        0
    \end{pmatrix}
    =: \mathbf{x}.
\end{equation*}
Define $\Omega_0:= \{z \neq 0\} \cap \{ \mathsf{x}_{2,T} \rightarrow z^2\}$, a probability 1 event under Assumptions \ref{ass:var}-\ref{ass:z_0}. To avoid division-by-zero issues, the rest of the proof is on $\Omega_0$ for large $T$, ensuring $z_T, \mathsf{x}_{2,T}$ and $z$ are non-zero. \\
\noindent \emph{Proof of (\ref{eq:white_dist_noy}):} Derive the identity
\[
    S_T^{(\rho)} := \frac{\rho^T}{\rho^2-1} (\hat{\rho} - \rho) - \frac{y_T}{z}  = \frac{\beta^{T-2} A_T - y_T z + y_T z}{(1-\beta^2)\beta^{2(T-2)}B_T} - \frac{y_T z}{z^2}.
\]
Recognizing that many of the terms in $S_T^{(\rho)}$ are elements of the vector $\mathbf{x}_T$, take the absolute value, and apply the triangle inequality, delivering 
\[
|S_T^{(\rho)}|=\left| \tfrac{\mathsf{x}_{3,T}+y_T z}{\mathsf{x}_{2,T}} - \tfrac{y_T z}{z^2} \right|
\leq \tfrac{ |\mathsf{x}_{3,T} | }{|\mathsf{x}_{2,T}|} + |y_T| |z| \left|  \tfrac{1}{\mathsf{x}_{2,T}} -  \tfrac{1}{z^2}  \right|.
\]
For $a, b^2 > 0$ one has the identity
\[
|\tfrac{1}{a}-\tfrac{1}{b^2}|= \tfrac{|a-b^2|}{|ab^2|} = \tfrac{\left|\sqrt{a} - |b| \right|(\sqrt{a}+|b|)}{|ab^2|}.
\]
Combine the identity with the previous bound. Subsequently, noting $|ab|=|a||b|$, cancel the $|z|$ in the numerator and denominator. This delivers
\[
S_T^{(\rho)} \leq \tfrac{ |\mathsf{x}_{3,T} | }{|\mathsf{x}_{2,T}|} + \tfrac{|z| |y_T| \left| \sqrt{\mathsf{x}_{2,T}} -|z| \right| (\sqrt{\mathsf{x}_{2,T}} +|z|)}{|z^2 \mathsf{x}_{2,T}|} = \tfrac{ |\mathsf{x}_{3,T} | }{|\mathsf{x}_{2,T}|} + \tfrac{\mathsf{x}_{4,T} (\sqrt{\mathsf{x}_{2,T}} +|z|)}{|z\mathsf{x}_{2,T}|} =: g_1(\mathbf{x}_T, z).
\]
As we are working on $\Omega_0$, the candidate limit $g_1(\mathbf{x},z)$ is continuous. $g_1(\mathbf{x}_T,z) \overset{a.s.}{\rightarrow} g_1(\mathbf{x},z) = 0$ follows by the continuous mapping theorem. \\
\leavevmode \\
\emph{Proof of (\ref{eq:and_est_noy}):} Define the absolute difference
\[
S_T^{(A)} := \left| \frac{A_T}{\sqrt{B_T}} - Sign(z) \sqrt{1-\beta^2} y_T \right| .
\]
On $\Omega_0$,  $Sign(z) = z/|z|$. Use  $Sign(z) = z/|z|$, factor out $\sqrt{1-\beta^2}$, extend by $\beta^{T-2}$, then add and subtract $y_T z$. This yields
\begin{equation*}
    S_T^{(A)} = \sqrt{1-\beta^2} \left| \frac{ \beta^{T-2}A_T - y_T z + y_T z}{\sqrt{(1-\beta^2) \beta^{2(T-2)} B_T}} - \frac{ y_T z }{|z|} \right|.
\end{equation*}
This can be re-written as elements of the vector $\mathbf{x}_T$, after which the triangle inequality is applied. This yields
\begin{equation*}
     S_T^{(A)} = \sqrt{1-\beta^2} \left| \tfrac{\mathsf{x}_{3,T}}{\sqrt{\mathsf{x}_{2,T}}} + \tfrac{y_T z}{\sqrt{\mathsf{x}_{2,T}}} - \tfrac{y_T z}{|z|} \right| \leq \sqrt{1-\beta^2} \left(  \left| \tfrac{\mathsf{x}_{3,T}}{\sqrt{\mathsf{x}_{2,T}}} \right| + |z| |y_T|\left| \tfrac{1}{\sqrt{\mathsf{x}_{2,T}}} - \tfrac{1}{|z|} \right| \right).
\end{equation*}
See that
\begin{equation*}
|z| |y_T|\left| \tfrac{1}{\sqrt{\mathsf{x}_{2,T}}} - \tfrac{1}{|z|}  \right| = |z| |y_T| \left| \tfrac{|z| - \sqrt{\mathsf{x}_{2,T}}}{|z|  \sqrt{\mathsf{x}_{2,T}}} \right| = \tfrac{|z| |y_T| \left|\sqrt{\mathsf{x}_{2,T}} - |z| \right|}{|z|  \sqrt{\mathsf{x}_{2,T}}} =  \tfrac{ \mathsf{x}_{4,T}}{  \sqrt{\mathsf{x}_{2,T}}}.
\end{equation*}
From which it follows that
\begin{equation*}
    S_T^{(A)} \leq \sqrt{1-\beta^2} \left(  \left| \tfrac{\mathsf{x}_{3,T}}{\sqrt{\mathsf{x}_{2,T}}} \right| + \tfrac{ \mathsf{x}_{4,T}}{  \sqrt{\mathsf{x}_{2,T}}} \right) =: g_2(\mathbf{x}_{T},z).
\end{equation*}
On $\Omega_0$, the mapping $(u,v) \mapsto g_2(u,v)$ is continuous at $(\mathbf{x},z)$. Therefore, by the continuous mapping theorem, $g_2(\mathbf{x}_T,z) \overset{a.s.}{\rightarrow} g_2(\mathbf{x},z) = 0$. \\
\leavevmode \\
\emph{Proof of (\ref{eq:OLS_consistency})}: Define $\mathsf{w}_T := f_T \beta^{T-2}A_T$. By Lemma \ref{lem:num_conv_zero}, using Assumptions \ref{ass:var}-\ref{ass:z_0}, then $\mathsf{w}_T \overset{\mathbb{P}}{\rightarrow} 0$ or in the case where $|f_T|$ summable, $\mathsf{w}_T \overset{a.s.}{\rightarrow} 0$.
Derive the identity
\begin{equation*}
    f_T \rho^{T-2} (\hat{\rho} - \rho) = \tfrac{f_T \beta^{T-2}A_T}{\beta^{2(T-2)}B_T} = \tfrac{\mathsf{w}_{T}}{(1-\beta^2)\mathsf{x}_{2,T}} =: g_3(\mathbf{x}_T, \mathsf{w}_T).
\end{equation*}
For fixed $\mathsf{w}_T$, the mapping $v \mapsto g_3(v,\mathsf{w}_T)$ is continuous at $\mathbf{x}$ on $\Omega_0$. Hence, the continuous mapping theorem is applicable and $g(\mathbf{x}_T, \mathsf{w}_T) {\overset{\mathbb{P}}{\rightarrow}} g(\mathbf{x},0)= 0.$ If $\sum_{T=1}^\infty|f_t| < \infty$, then $
g_3(\mathbf{x}_T, \mathsf{w}_T) {\overset{a.s.}{\rightarrow}} g_3(\mathbf{x},0)= 0._{\qedsymbol{}}$

\subsubsection{Proof of Corollary \ref{cor:consistency}}
    Choose $f_T = \rho^{-T+2}$, which vanishes and is absolutely summable. Apply Theorem \ref{thm:consistency}.

\begin{lemma} \label{lem:a.s. results}
    Under Assumption \ref{ass:var}
    \begin{align*}
        (1-\beta^2)\beta^{2(T-2)} B_T - z_T^2 \overset{a.s.}{\rightarrow}& \, 0 \\ 
        \beta^{T-2}A_T - z_T y_T \overset{a.s.}{\rightarrow}& \, 0
    \end{align*}
\end{lemma}
\noindent \textbf{Proof.} Under Assumption \ref{ass:var}, Lemma \ref{lem:z_T_bounds} and Lemma \ref{lem:z_bounds} apply. Therefore, there exists a $K$ such that
\begin{equation*}
    \mathbb{E}|y_T| \leq K,\quad \mathbb{E}[(z_T - z)^2] \leq K \beta^{2T}.
\end{equation*}
Utilizing Cauchy-Schwarz and the aforementioned bounds, one has that
\begin{equation*}
    \mathbb{E} |z_T y_T - z y_T| = \mathbb{E} (|y_T| |z_T - z|) \leq  \left\{ \mathbb{E} y_T^2 \mathbb{E} (z_T - z)^2  \right\}^{1/2} \leq \left\{ K^2 \beta^{2T}  \right\}^{1/2} = K |\beta|^T.
\end{equation*}
This bound is summable. Furthermore, under Assumption \ref{ass:var}, by Lemma \ref{lem:z_bounds}, there exists a finite $K$, such that $ \quad \mathbb{E}|z_T^2 - z^2 | \leq K|\beta|^{T}$, which is also summable. It follows by Lemma \ref{lem:reduction_lemma} that
\begin{equation} \label{eq:as_zero}
     |z_T^2 - z^2| \overset{a.s.}{\rightarrow}0, \quad |z_T y_T - z y_T| \overset{a.s.}{\rightarrow} 0.
\end{equation}
Via the triangle inequality one can derive the bounds
\begin{align*}
    |(1-\beta^2)\beta^{2(T-2)} B_T - z_T^2| &\leq  |(1-\beta^2)\beta^{2(T-2)} B_T - z^2| + | z_T^2 - z^2| \\
    |\beta^{T-2}A_T - y_Tz_T| &\leq |\beta^{T-2}A_T - y_Tz | + |y_Tz - y_Tz_T|.
\end{align*}
The first summand in each bound converges to zero by Theorem \ref{thm:consistency}, requiring Assumption \ref{ass:var}, whilst the second summand in each bound converges by (\ref{eq:as_zero})$._{\qedsymbol_{}}$

\subsubsection{Proof of Lemma \ref{lem:numden_joint}} \label{app:numden_joint}
Under Assumption \ref{ass:y_exists}, $(y_T, z_T) \overset{D}{\rightarrow} (y,z)$. Additionally, under Assumption \ref{ass:var}, Lemma \ref{lem:a.s. results} applies. Therefore, by Theorem 4.4 of \cite{Bil68}, the two results can be combined as
\begin{equation*}
\mathbf{y}_T :=
    \begin{pmatrix}
       y_T \\ z_T \\  ({1-\beta^2}) \beta^{2(T-2)} B_T - {z_T^2} \\ \beta^{T-2} A_T - z_T y_T
    \end{pmatrix}
    \overset{D}{\rightarrow} 
        \begin{pmatrix}
        y \\ z \\ 0 \\ 0
    \end{pmatrix} := \mathbf{y}.
\end{equation*}
The mapping $g: \mathbb{R}^4 \rightarrow \mathbb{R}^2$, $
(x_1, x_2, x_3, x_4) \mapsto g(\mathbf{x}) = (x_4 + x_2 x_1, x_3 + x_2^2) $ is continuous. An application of the continuous mapping theorem yields
\begin{equation*}
    (
        \beta^{T-2} A_T,  (1-\beta^2) \beta^{2(T-2)} B_T
    )
    = g(\mathbf{y}_T) \overset{D}{\rightarrow} g(\mathbf{y})
     = (
         zy , z^2
    )._{\qedsymbol{}}
\end{equation*}

\subsubsection{Proof of Theorem \ref{thm:and_dist}}
Under Assumptions \ref{ass:var}, \ref{ass:y_exists},  Lemma \ref{lem:numden_joint} is applicable and the terms $A_T, B_T$ converge jointly in distribution. We combine this joint convergence with the continuous mapping theorem. Once again, define $\Omega_0:= \{z \neq 0\} \cap \{ B_T \rightarrow (1-\beta^2)^{-1}z^2\}$, a probability 1 event under Assumptions \ref{ass:var}-\ref{ass:z_0}. We work on $\Omega_0$ with large $T$ to resolve division-by-zero issues.\\
\textbf{Proof of (\ref{eq:white_dist}):} Define $g (x_1, x_2) :=  x_1/x_2$. By the continuous mapping theorem
\begin{equation*}
    \frac{\rho^T A_T}{(\rho^2-1) B_T} = \frac{\beta^{T-2}A_T}{(1-\beta^2)\beta^{2(T-2)}B_T} = g \begin{pmatrix} \beta^{T-2} A_T \\ (1-\beta^2) \beta^{2(T-2)} B_T \end{pmatrix} \overset{D}{\rightarrow} g \begin{pmatrix}
        yz \\ z^2
    \end{pmatrix} = \frac{y}{z} .
\end{equation*}
\textbf{Proof of (\ref{eq:and_est})}: Define $h(x_1, x_2) := x_1/\sqrt{x_2}$. By the continuous mapping theorem
    \begin{align*}
        \frac{A_T}{\sqrt{B_T}} = \sqrt{1-\beta^2} h \begin{pmatrix}
            \beta^{T-2} A_T \\ (1-\beta^2) \beta^{2(T-2)} B_T
        \end{pmatrix} \overset{D}{\rightarrow}  \sqrt{1-\beta^2} h \begin{pmatrix}
            yz \\ z^2
        \end{pmatrix} .
    \end{align*}
On $\Omega_0$, ${z}/{|z|} = Sign(z)$. Hence $ h (yz,z^2) = y(z/{|z|}) \overset{a.s.}{=} Sign(z) \, y._{\qedsymbol{}}$

\subsubsection{Proof of Corollary \ref{cor:symm}}
The property $ (\epsilon_1, \epsilon_2, ..., \epsilon_T) \overset{d}{=} - (\epsilon_1, \epsilon_2, ..., \epsilon_T) $  immediately implies $y_T \overset{D}{=} -y_T$, as $y_T$ is an odd linear function of $(\epsilon_1, \epsilon_2, ..., \epsilon_T)$. \\ \leavevmode \\
Next, $y_T \overset{D}{\rightarrow} y$, and by the continuous mapping theorem $-y_T \overset{D}{\rightarrow} -y$. Since $y_T \overset{D}{=} -y_T$ for all $T$, then $y \overset{D}{=} -y$. \\ 
\leavevmode \\
We now demonstrate that $Sign(z)y \overset{D}{=} y$. Under Assumption \ref{ass:z_0}, $\{ z=0 \}$ is a null set. Therefore, for any $ A \in \mathcal{B}(\mathbb{R})$, by the law of total probability
\begin{align*}
    \mathbb{P} \left(  Sign(z) y \in A \right) 
    & = \mathbb{P} \left(  y \in A | \{ z > 0 \} \right) \mathbb{P} (\{ z > 0 \}) \\ & \hspace{12pt} + \mathbb{P} \left( - \, y \in A | \{z < 0 \} \right) \mathbb{P} (\{z < 0 \} ).
\end{align*}
By independence of $y$ and $z$, it follows that
\begin{equation*}
    \mathbb{P} \left(  Sign(z) y \in A \right) 
    = \mathbb{P} \left(  y \in A \right) \mathbb{P} (\{ z > 0 \}) + \mathbb{P} \left( - \, y \in A \right) \mathbb{P} (\{z < 0 \} ).
\end{equation*}
By symmetry of $y$ and $\mathbb{P}(z \notin \{0\})=1$, requiring Assumption \ref{ass:z_0}, one has that
\begin{equation*}
    \mathbb{P} \left(  Sign(z) y \in A \right) = \mathbb{P} \left(  y \in A \right).
\end{equation*}
Combine $Sign(z)y \overset{D}{=} y$ with  Theorem \ref{thm:and_dist}, requiring  Assumptions \ref{ass:var}-\ref{ass:y_exists}, to see that:
\begin{equation*}
    \frac{A_T}{\sqrt{B_T}} \overset{D}{\rightarrow} \sqrt{1-\beta} \, Sign(z) y \overset{D}{=} \sqrt{1-\beta} \, y._{\qedsymbol{}}
\end{equation*}

\subsection{Proofs of Section \ref{sec:yz}}
Below we prove the properties of $z,y$ presented in the second half of the paper.

\subsubsection{Proof of Lemma \ref{lem:SV}}
    Define $\mathcal{G} = \{ \sigma_1^2, \sigma_2^2, ...\}$ and recall the lemma assumes $x_0 \overset{a.s.}{=} 0$. Next, as the innovations satisfy $\sup_{t>0} \mathbb{E} \epsilon_t<\infty$, by Lemma \ref{lem:cauchy}, one has $z \overset{a.s.}{=} \sum_{t=1}^{\infty} \beta^{t-1} \epsilon_t$. Furthermore, as $\epsilon_t = \zeta_t \sigma_t^2$ with $\zeta_t \overset{iid}{\sim}N(0,1)$ with $\sigma(\sigma_t^2:t>0)$ and $\sigma(\zeta_t: t>0)$ independent, one has
    \begin{equation*}
        z|\mathcal{G} \sim N \Big(0, \sum_{t=1}^\infty \beta^{2(t-1)} \sigma_t^2 \Big).
    \end{equation*}
    Where the variance of the Gaussian is finite $a.s.$, as  by the monotone convergence theorem, one has
    \begin{equation*} 
        \mathbb{E} \Big[ \sum_{t=1}^\infty \beta^{2(t-1)} \sigma_t^2 \Big] = \sum_{t=1}^\infty \beta^{2(t-1)} \mathbb{E} [\sigma_t^2] \leq \frac{M^2}{1-\beta^2}.
    \end{equation*}
    The variance is also nonzero, as for any fixed $t$, $ \sigma_t^2>0$ $a.s.$ was assumed. By the law of iterated expectations, we then have
    \begin{equation*}
        \mathbb{P}(z \in \{0\}) = \mathbb{E} [\mathbb{P}(z \in \{0\}|\mathcal{G})] = 0._{\qedsymbol{}}
    \end{equation*}

\subsubsection{Proof of Lemma \ref{lem:ARMA_znonzero}}
Assumption \ref{ass:arma} states that the innovations are a stable Gaussian ARMA process. Therefore, by \cite[Theorem 3.1.1.,]{BroDav91}, requiring $\varphi(\mathsf{z}) \neq 0$ for $|z| \leq 0$, each $\epsilon_t$ admits an MA($\infty$) representation of the form
\begin{equation*}
    \epsilon_t = \sum_{j=0}^{\infty} \psi_j \zeta_{t-j}, \quad \text{ where} \quad \psi(\mathsf{z})= \sum_{j=0}^{\infty} \psi_j \mathsf{z}^j = \frac{\theta(\mathsf{z})}{\varphi(\mathsf{z})}.
\end{equation*}
Recalling that $x_0 \overset{a.s.}{=} 0$, one therefore has
\begin{equation*}
    z \overset{a.s.}{=} \sum_{t=1}^{\infty} \beta^{t-1} \epsilon_t =  \sum_{t=1}^{\infty} \beta^{t-1} \sum_{j=0}^{\infty} \psi_j \zeta_{t-j} = \zeta_0\sum_{h=0}^{\infty} \beta^{h} \psi_h + R.
\end{equation*}
Where in the final equality $\zeta_0\sum_{h=0}^{\infty} \beta^{h} \psi_h$ collects all terms containing $\zeta_0$, whilst $R$ collects all terms containing $(..., \zeta_{-1}, \zeta_1, ...)$. We wish to show that $\zeta_0$ exactly cancels all other terms with probability zero. Next, note that the leading term of $\zeta_0$ takes on the form $\sum_{h=0}^{\infty} \beta^{h} \psi_h = \psi(\beta) =  \frac{\theta(\mathsf{\beta})}{\varphi(\beta)}$. This term cannot be zero, as the polynomial in the numerator $\theta(\mathsf{z})$ has no roots for $|z|<1$. Furthermore, as $\varphi(\mathsf{z})\neq 0$ for $|z|<1$, there are no division-by-zero issues. For brevity, let $c=\sum_{h=0}^{\infty} \beta^{h} \psi_h$, by the law of iterated expectations, one has
\begin{equation*}
    \mathbb{P}(z\in \{0\}) = \mathbb{P}(\zeta_0 \in \{R/c\}) = \mathbb{E}[\mathbb{P}(\zeta_0 \in \{R/c\} | R)] = 0.
\end{equation*}
The final equality uses that $R$ and $\zeta_0$ are independent and that $\zeta_0$ is non-atomic$._{\qedsymbol{}}$

\subsubsection{Proof of Lemma \ref{lem:y_dist_stat}}
If we can demonstrate existence of $y$, a second moment of $y$ follows by  Lemma \ref{lem:finite_l2}, requiring Assumption \ref{ass:var} and existence of $y$. The rest of the proof therefore focuses on proving the existence of $y$ under stationarity and Assumption \ref{ass:var}.\\ 
\leavevmode \\
By Proposition 6.5 of \cite{Brei68}, for any single-ended stationary process $(\epsilon_t)_{t \in \mathbb{N}}$, there exists a double-ended stationary process $(\tilde{\epsilon}_t)_{t \in \mathbb{Z}}$ such that $(\epsilon_t)_{t \in \mathbb{N}} \overset{D}{=}  (\tilde{\epsilon}_t)_{t \in \mathbb{N}}$. Therefore, by stationarity of $(\tilde{\epsilon}_t)_{t \in \mathbb{Z}}$ one has that, for any $T$
\begin{equation*}
    (\epsilon_1, \epsilon_2, ..., \epsilon_T) \overset{D}{=}
    (\tilde{\epsilon}_1, \tilde{\epsilon}_2, ..., \tilde{\epsilon}_T) \overset{D}{=} (\tilde{\epsilon}_{-(T-1)}, \tilde{\epsilon}_{-(T-2)}, ..., \tilde{\epsilon}_0).
\end{equation*}
The variable $y_T = \sum_{t=1}^T \beta^{T-t} \epsilon_t$ is a linear map of $(\epsilon_t)_{t=1}^T$. It follows that
\begin{equation*}
    y_T = \sum_{t=1}^T \beta^{T-t} \epsilon_t \overset{D}{=} \sum_{t=1}^T \beta^{T-t} \tilde{\epsilon}_t \overset{D}{=} \sum_{t=0}^{T-1} \beta^t \tilde{\epsilon}_{-t} =: \tilde{y}_T.
\end{equation*}
For $k \in \mathbb{N}$, define $\tilde{\bm{\epsilon}}_k :=(..., \tilde{\epsilon}_{k-1}, \tilde{\epsilon}_{k})$. Furthermore, define the functional $\varphi$ via
\begin{equation*}
    \varphi \, \circ \, \tilde{\bm{\epsilon}}_k  := \sum_{t=0}^\infty \beta^t \tilde{\epsilon}_{k-t}.
\end{equation*}
\emph{For the rest of the proof, almost sure convergence $\overset{a.s.}{\rightarrow}$ and expectations $\mathbb{E}[\cdot]$ are taken over the measure induced by $(\tilde{\epsilon}_t)_{t \in \mathbb{Z}}.$} \\ \leavevmode
\\
Picking the candidate probability limit $y := \varphi \, \circ \, \tilde{\bm{\epsilon}}_0$, by use of the triangle inequality, Monotone Convergence Theorem and the Cauchy-Schwarz bound $\mathbb{E}|\tilde{\epsilon}_T| \leq M$ under Assumption \ref{ass:var}, one can derive
\begin{equation*}
    \mathbb{E} |\tilde{y}_T - y| = \mathbb{E} \bigg| \sum_{t=T}^\infty \beta^t \tilde{\epsilon}_{-t} \bigg| {\leq} \mathbb{E} \sum_{t=T}^\infty |\beta|^t |\tilde{\epsilon}_{-t}| =  \sum_{t=T}^\infty |\beta|^t \mathbb{E} |\tilde{\epsilon}_{-t}| = \frac{M |\beta|^{T}}{1-|\beta|}.
\end{equation*}
This bound is summable, and so by Lemma \ref{lem:reduction_lemma}, $\tilde{y}_T \overset{a.s.}{\rightarrow} y$. It follows that $y_T \overset{D}{=} \tilde{y}_T \overset{a.s.}{\rightarrow} y._{\qedsymbol{}}$
\begin{definition}
     Drawing inspiration from a truncation argument in \cite{And1959}, let $\lfloor \cdot \rfloor$ denote the floor function and define
    \begin{align*}
        z_T^* &:= \rho x_0 + \sum_{t=1}^{\lfloor\frac{1}{3}T\rfloor} \beta^{t-1} \epsilon_t, \quad 
        \tilde{z}_T := \sum_{t=\lfloor\frac{1}{3}T\rfloor+1}^{T} \beta^{t-1} \epsilon_t \\
        y_T^* &:= \sum_{t=\lfloor\frac{2}{3}T\rfloor+1}^{T} \beta^{T-t} \epsilon_t, \hspace{22pt} \tilde{y}_T := \sum_{t=1}^{\lfloor \frac{2}{3}T\rfloor} \beta^{T-t} \epsilon_t.
    \end{align*}
\end{definition}
   
\begin{lemma} \label{lem:as_0}
    Under Assumption \ref{ass:var}, $z_T^*\overset{a.s.}{\rightarrow}z$, $\tilde{z}_T \overset{a.s.}{\rightarrow} 0$ and $\tilde{y}_T \overset{a.s.}{\rightarrow} 0$. If, additionally, $y_T \overset{D}{\rightarrow} y$, then $y_T^* \overset{D}{\rightarrow} y$.
\end{lemma}
\noindent \textbf{Proof:} $z_T^* \overset{a.s.}{\rightarrow}z$ follows by noting $z_T^* = z_{\lfloor \frac{T}{3}\rfloor +1} \overset{a.s.}{\rightarrow} z$ by similar arguments to Lemma \ref{lem:measurability}, requiring Assumption \ref{ass:var}. $\tilde{z}_T \overset{a.s.}{\rightarrow} 0 $ follows by the continuous mapping theorem as $(z_T, z_T^*) \overset{a.s.}{\rightarrow} (z,z)$ and $\tilde{z}_T = z_T - z_T^*$. \\
\leavevmode \\
To prove $\tilde{y}_T \overset{a.s.}{\rightarrow} 0$, see that via the triangle inequality and Assumption \ref{ass:var}
\begin{equation*} \textstyle
    \mathbb{E}|\tilde{y}_T| = \mathbb{E}| \sum_{t=1}^{\lfloor \frac{2}{3}T  \rfloor} \beta^{T-t} \epsilon_t| \leq M  \sum_{t=1}^{\lfloor \frac{2}{3}T  \rfloor} |\beta|^{T-t} \leq \frac{M}{1-|\beta|} \beta^{T/3}.
\end{equation*}
By Lemma \ref{lem:reduction_lemma}, the moment bound implies $\tilde{y}_T \overset{a.s.}{\rightarrow} 0$. Finally, if $y_T \overset{D}{\rightarrow} y$, then by Theorem 4.4 of \cite{Bil68}, $(y_T,\tilde{y}_T) \overset{D}{\rightarrow} (y,0)$ and $y_T^* = y_T - \tilde{y}_T \overset{D}{\rightarrow} y$ follows by the continuous mapping theorem$._{\qedsymbol{}}$

\subsubsection{Proof of Theorem \ref{lem:independence}}
Define $F_{T\, (y^*, z^*)}$ as the joint CDF of $y_T^*, z_T^*$, and define $F_{T\, (y^*)}, F_{T\,(z^*)}$ as the marginals. Similarly, define $F_y, F_z$ as the marginal CDFs of $y,z$. \\ \leavevmode \\
For any $s,u \in \mathbb{R}$, via the triangle inequality, one has that
\begin{equation} \label{eq:cdfs}
 \left| F_{T\, (y^*, z^*)} - F_{y}F_z \right| \leq \left| F_{T\, (y^*, z^*)} - F_{T \, (y^*)}F_{T \, (z^*)} \right| + \left| F_{T \, (y^*)}F_{T \, (z^*)} - F_{y}F_z \right|,
\end{equation}
where joint CDFs are evaluated at $(s,u)$ and marginals at $s$ and $u$, respectively. \\ \leavevmode
\\
The first term in (\ref{eq:cdfs}) vanishes as $\{ y_T^* \leq s \} \in \sigma(\epsilon_t: t \geq \lfloor \tfrac{2}{3}T \rfloor + 1)$ and $\{ z_T^* \leq u \} \in \sigma(x_0, \epsilon_t: t \leq \lfloor \tfrac{1}{3}T \rfloor)$. As the two $\sigma$-algebras are spaced $\left( \lfloor \tfrac{2}{3}T \rfloor + 1 - \lfloor \tfrac{T}{3} \rfloor \right)$ apart, one can bound the first term by the $\alpha$-mixing coefficient, which vanishes as $T \rightarrow \infty$ under Assumption \ref{ass:alpha_mixing}. The second term converges to zero, as $y_T^* \overset{D}{\rightarrow}y$ and $z_T^* \overset{D}{\rightarrow} z$ under Assumptions \ref{ass:var} and \ref{ass:y_exists} by Lemma \ref{lem:as_0}. This establishes independence of $y,z$.\\ \leavevmode
\\
To conclude the joint convergence result, recall that by Lemma \ref{lem:as_0}, under Assumptions \ref{ass:var} and \ref{ass:y_exists}, one has $(\tilde{y}_T, \tilde{z}_T, y_T^*, z_T^*) \overset{D}{\rightarrow} (0, 0, y, z)$. This can be combined with the earlier joint convergence result of $(y_T^*, z^*_T)$ to establish convergence via the continuous mapping theorem, yielding
\begin{equation*}
    (y_T, z_T) = (y_T^* + \tilde{y}_T, z^*_T + \tilde{z}_T) = (y_T^*, z_T^*) + (\tilde{y}_T, \tilde{z}_T) \overset{D}{\rightarrow} (y,z).
\end{equation*} 
$z$ has a finite second moment by (\ref{eq:zsq_bound}) of Lemma \ref{lem:z_bounds}, requiring Assumption \ref{ass:var}. $y$ has a finite second moment by Lemma \ref{lem:finite_l2}, requiring the limiting distribution $y$ to exist and Assumption \ref{ass:var}$._{\qedsymbol{}}$

\subsection{ARMA Proofs}

\subsubsection{Proof of Lemma \ref{lem:arma}}
For innovations satisfying Assumption \ref{ass:arma}, we have that $x_0 \overset{a.s.}{=} 0$ and $y_T, z_T$ are sums of mean-zero Gaussians. Therefore, both objects have a Gaussian distribution around zero. To characterize the covariance matrix of the joint distribution, we begin by deriving the diagonal values. The diagonals converge to $\sigma^2$ as $T \rightarrow \infty$ as
\begin{equation*}
    \sigma_T^2 := var(y_T) = var(z_{T+1}) = \sum_{i=0}^{T-1} \sum_{j=0}^{T-1} \beta^{i+j} \gamma_{|i-j|} \rightarrow \frac{1}{1-\beta^2} \left[ \gamma_0 + 2 \sum_{h=1}^\infty \gamma_h \beta^h \right] = \sigma^2.
\end{equation*}
Stationary Gaussian ARMA(p,q) models with roots outside the unit circle as in Assumption \ref{ass:arma} are $\alpha$-mixing by \cite[Theorem 6, p.99]{Dou1994}. Therefore, by Theorem \ref{lem:independence}, the limits $y, z$ are independent, hence the limiting covariance matrix of $(y_T, z_T)$ is given by $\sigma^2 I_2$$._{\qedsymbol{}}$

\subsubsection{Proof of Proposition \ref{thm:arma}}
This follows from combining the joint distribution of $(y,z)$ provided in Lemma \ref{lem:arma} with Corollary \ref{cor:symm}. The symmetry condition required by the corollary is satisfied here as under Assumption \ref{ass:arma}, any block of innovations $(\epsilon_t, ..., \epsilon_{t+h})$ is Gaussian with symmetric covariance matrix and mean zero.

\subsubsection{Proof of Lemma \ref{lem:uniformbeta}}
Recall $\hat{\beta} := \hat{\rho}^{-1}$ with $\hat{\rho} = \rho + A_T/B_T$. It then follows that for any $h \in \mathbb{N}_0$
\begin{equation} \label{eq:identity}
    \hat{\beta}^h = \beta^h \Big( 1+ \frac{\beta A_T}{B_T} \Big)^{-h}.
\end{equation}
Apply $\exp(\ln(x))=x$ to the multiplicative difference term, then apply a Taylor expansion of $ln(x)$, delivering
\begin{equation*}
    \Big( 1+ \frac{\beta A_T}{B_T} \Big)^{-h} = \exp \Big\{ -h \ln(1 + \beta A_T/ B_T)\Big \} = \exp\Big\{ -h \sum_{j=1}^\infty (-1)^{j-1} (\beta A_T/ B_T)^{j}/j \Big\}.
\end{equation*}
We wish to bound $s_{T,h} := -h \ln(1 + \beta A_T/ B_T)$. Begin by bounding via the triangle inequality and the bounds $(-1)^{j-1} \leq 1$ and $ j^{-1} \leq 1$, delivering
\begin{equation} \label{eq:s_bound1}
    |s_{T,h}| = \Big| -h  \sum_{j=1}^\infty (-1)^{j-1} (\beta A_T/ B_T)^{j}/j \Big| \leq h \sum_{j=1}^{\infty} |\beta A_T/ B_T|^{j}. 
\end{equation}
Recall that $\beta^{3(T-2)/2} A_T \overset{a.s.}{\rightarrow} 0$ by Lemma \ref{lem:num_conv_zero} and $\beta^{2(T-2)} B_T \overset{a.s.}{\rightarrow} (1-\beta^2)^{-1} z^2$ by Lemma \ref{lem:B_T}, both requiring Assumption \ref{ass:var}. By Assumption \ref{ass:z_0}, $\{ z=0\}$ is a null set. Hence, one may pick a $\delta$ satisfying $0 < \delta <1 $ and define the probability 1 event $\Omega_\delta := \{\lim_T |\frac{\beta^{3(T-2)/2} A_T}{ \beta^{2(T-2)} B_T}| < \delta\} $. On $\Omega_\delta$, for sufficiently large $T$, one has for $j \in \mathbb{N}$
\begin{equation} \label{eq:s_bound2}
    \Big| \frac{\beta A_T}{B_T}  \Big|^j =  \Big|  \frac{\beta \beta^{2(T-2)} A_T}{ \beta^{2(T-2)}B_T}  \Big|^j =  |\beta|^{jT/2}  \Big| \frac{\beta^{3(T-2)/2} A_T}{ \beta^{2(T-2)} B_T}  \Big|^j \leq |\beta|^{T/2}  \delta^j.
\end{equation}
Where in the final inequality, we used $|\beta|^{T/2}<1$, so $|\beta|^{jT/2}< |\beta|^{T/2}$. To bound $s_{T,h}$ uniformly on the set $\Omega_\delta$ for large $T$, apply (\ref{eq:s_bound1}), then (\ref{eq:s_bound2}) and solve the geometric series. This delivers
\begin{equation} \label{eq:s_bound3}
    \sup_{0\leq h \leq T} | s_{T,h} | \leq T \sum_{j=1}^{\infty} |\beta A_T/ B_T|^{j} \leq  T |\beta|^{T/2}\sum_{j=1}^{\infty} \delta^{j} \leq \frac{T |\beta|^{T/2}}{1-\delta} =: g_T.
\end{equation}
$|g_T|=g_T$ converges to zero as $\beta^{T/2}T \rightarrow 0$. Next, by the identity (\ref{eq:identity}), one has
\begin{equation*}
    \sup_{0<h\leq T} | \hat{\beta}^h - \beta^h| = \sup_{0<h\leq T} \Big| \beta^h - \beta^h \Big( 1+ \frac{\beta A_T}{B_T} \Big)^{-h} \Big| \leq \sup_{0<h\leq T}  \Big| 1 - \Big( 1+ \frac{\beta A_T}{B_T} \Big)^{-h} \Big| .
\end{equation*}
Recalling $(1+\beta A_T/B_T)^{-h} = \exp \{ s_{T,h}\}$, via the inequality $|1-e^x| \leq |x| e^{|x|}$, one has
\begin{equation*}
    \sup_{0\leq h \leq T} | \hat{\beta}^h-\beta^h| \leq \sup_{0\leq h \leq T}  \Big| 1 - \exp \{s_{T,h}\}\Big| \leq \sup_{0\leq h \leq T}  |s_{T,h}| \exp \{|s_{T,h}|\}.
\end{equation*}
By (\ref{eq:s_bound3}), and as $x \mapsto |x| \exp\{|x|\}$ is increasing for $x>0$, on the set $\Omega_\delta$ for sufficiently large $T$
\begin{equation*}
    \sup_{0\leq h \leq T} | \hat{\beta}^h-\beta^h| \leq  |g_T| e^{|g_T|}.
\end{equation*}
Which vanishes as $T \rightarrow \infty$, proving the first statement. To prove the second statement, note that also $T g_T $ vanishes, as $T^{2} \beta^{T/2} \rightarrow 0$, so on the same set $\Omega_\delta$, for sufficiently large $T$, one has
\begin{equation*} \textstyle
    \sum_{h=1}^T |\hat{\beta}^h-\beta^h| \leq T |g_T| e^{|g_T|} \rightarrow 0._{\qedsymbol{}}
\end{equation*}

\subsubsection{Autocovariances}
Many of the results below are derived under Assumption \ref{ass:arma}. As Assumption \ref{ass:arma} implies Assumptions \ref{ass:var}-\ref{ass:indep}, proofs of statements requiring Assumption \ref{ass:arma} will invoke previous results of the paper without caveat.

\begin{lemma} \label{lem:uniform}
    Under Assumption \ref{ass:arma}, one has $T^{-1/2} y_T z \overset{a.s.}{\rightarrow} 0$.
\end{lemma}
\noindent \textbf{Proof.} Under Assumption \ref{ass:arma}, $y_T$ and $z$ are sums of mean-zero Gaussians and therefore Gaussian with mean zero. Furthermore, by Lemmas \ref{lem:z_T_bounds}-\ref{lem:z_bounds}, there exists a finite $K$ bounding $\mathbb{E}[y_T^2], \mathbb{E}[z^2]$, uniformly in $T$. Recall that that for $z, y_T$ Gaussian with a second moment with upper bound $K$, $\mathbb{E} y_T^6 \leq 15K^3$ and $\mathbb{E} z^6 \leq 15K^3$. Therefore, by the Cauchy-Schwarz inequality and the sixth moment bound for Gaussians, one has
\begin{equation*}
    \mathbb{E} \Big( |T^{-1/2} y_T z|^3 \Big) = T^{-3/2} \mathbb{E} \Big( |y_T z|^3 \Big) \leq T^{-3/2} \Big(\mathbb{E} y_T^6 \mathbb{E}z^6\Big)^{1/2} \leq T^{-3/2} 15K^3.
\end{equation*}
This bound is summable, hence by Lemma \ref{lem:reduction_lemma}, $(T^{-1/2}y_T z)^3 \overset{a.s.}{\rightarrow} 0._{\qedsymbol{}}$

\begin{lemma} \label{lem:momtounif}
    Let $(X_{T,h})$, $0 \leq h \leq L_T$, be a triangular array of random variables. Then
    \begin{align*}
        lim_T\sum_{h=0}^{L_T} \mathbb{E}|X_{T,h}| =& 0 \Rightarrow \sup_{0 \leq h\leq L_T} X_{T,h} \overset{\mathbb{P}}{\rightarrow}0.
    \end{align*}
\end{lemma}
\noindent \textbf{Proof:} For any $\varepsilon>0$, via Boole's inequality one has
\begin{equation*}
    \mathbb{P} \Big( \sup_{0 \leq h \leq L_T } |X_{T,h}|>\varepsilon \Big) = \mathbb{P} \Big( \bigcup_{h=0}^{L_T} \big\{ |X_{T,h}|>\varepsilon \big\} \Big) \leq \sum_{h=0}^{L_T} \mathbb{P}  \Big(  |X_{T,h}|>\varepsilon \Big). 
\end{equation*}
By Markov's inequality, one then has
\begin{equation*}
    \mathbb{P} \Big( \sup_{0 \leq h \leq L_T } |X_{T,h}|>\varepsilon \Big) \leq \varepsilon^{-1} \sum_{h=0}^{L_T} \mathbb{E} |X_{T,h}|.
\end{equation*}
Which vanishes as $T \rightarrow \infty$ by assumption$._{\qedsymbol{}}$

\begin{lemma} \label{lem:annoy}
    Under Assumption \ref{ass:arma}, then
    \begin{align}
        \sup_{0 \leq h \leq L_T} (T-h)^{-1/2}\beta^{T-h-2}A_{T-h} \overset{\mathbb{P}}{\rightarrow}& 0 \label{eq:annoy1} \\
        \sup_{0 \leq h \leq L_T} \beta^{T-2}A_h \overset{\mathbb{P}}{\rightarrow}& 0 \label{eq:annoy2}\\
        \sup_{0 \leq h \leq L_T} \beta^{2(T-2)} B_h \overset{\mathbb{P}}{\rightarrow}& 0. \label{eq:annoy3}
    \end{align}
\end{lemma}
\noindent \textbf{Proof:} To prove (\ref{eq:annoy1}), note that by adding and subtracting $y_T z$ and the triangle inequality, one has
\begin{equation*}
    |T^{-1/2} \beta^{T-2} A_T| \leq T^{-1/2} |\beta^{T-2} A_T - y_T z| + T^{-1/2} |y_{T-h}z| \overset{a.s.}{\rightarrow} 0.
\end{equation*}
Where the first term converges to $0$ by Theorem \ref{thm:consistency}, and the second term converges by Lemma \ref{lem:uniform}. Noting $ X_T \overset{a.s.}{\rightarrow} 0 \Rightarrow \sup_{0 \leq h\leq L_T} X_{T-h} \overset{\mathbb{P}}{\rightarrow} 0$, this implies (\ref{eq:annoy1}). \\ \leavevmode
\\
We now prove (\ref{eq:annoy2}). Let $M,K$ be finite constants such that $\mathbb{E}[\epsilon_t^2] \leq M, \mathbb{E}[z_T^2] \leq K$, the latter of which is guaranteed by Lemma \ref{lem:z_T_bounds}. Recalling that $\beta^{h-2}A_h = \sum_{t=1}^h \beta^{h-t} \epsilon_t z_t$, see that via the triangle and Cauchy-Schwarz inequalities
\begin{equation*}
    \mathbb{E}| \beta^{h-2} A_h | \leq \sum_{t=1}^h |\beta|^{h-t} \mathbb{E} \big[ |z_t \epsilon_t| ] \leq \sum_{t=1}^h |\beta|^{h-t} \sqrt{\mathbb{E} [z_t^2] \mathbb{E} [\epsilon_t^2]} \leq \frac{\sqrt{MK}}{1-|\beta|}.
\end{equation*}
Hence 
\begin{equation*}
    \sum_{h=0}^{L_T}\mathbb{E}| \beta^{T-2} A_h | = \sum_{h=0}^{L_T} |\beta|^{T-h} \mathbb{E} |\beta^{h-2} A_h|  \leq |\beta|^{T-L_T} \frac{\sqrt{MK}}{(1-|\beta|)^2} \rightarrow 0.
\end{equation*}
Which implies (\ref{eq:annoy2}) via Lemma \ref{lem:uniform}. \\
For $(\ref{eq:annoy3})$, recalling $\beta^{2(h-2)} B_h = \sum_{t=1}^h \beta^{2(h-t)} z_t^2$, one has
\begin{equation*}
    \mathbb{E}|\beta^{2(h-2)}B_h| \leq \sum_{t=1}^h \beta^{2(h-t)} \mathbb{E} [z_t^2] \leq \frac{K}{1-\beta^2}.
\end{equation*}
So
\begin{equation*}
    \sum_{h=0}^{L_T} \mathbb{E}[\beta^{2(T-2)} B_h] \leq \beta^{2(T-L_T)} \frac{K}{[1-\beta^2]^2} \rightarrow 0.
\end{equation*}
Which implies (\ref{eq:annoy3}) via Lemma \ref{lem:uniform}$._{\qedsymbol{}}$

\begin{lemma} \label{lem:uniformphi}
    Under Assumption \ref{ass:arma}, then 
    \begin{equation*} \textstyle
        \sup_{0 \leq h \leq L_T} |(T-h)^{-1} \sum_{t={h+1}}^T \epsilon_t \epsilon_{t-h} - \gamma_h | \overset{\mathbb{P}}{\rightarrow} 0.
    \end{equation*}
\end{lemma}
\noindent \textbf{Proof:} Denote $\hat{\phi}_{h} := (T-h)^{-1} \sum_{t=h+1}^T \epsilon_t \epsilon_{t-h}$. Then by \cite[Equation (7.3.5)]{BroDav91}, one has
\begin{equation*}
    \mathbb{E}[(\hat{\phi}_h- \gamma_h)^2 ] = (T-h)^{-1} \sum_{|k|<T-h} (1-|k|(T-h)^{-1})M_{k,h},
\end{equation*}
where $M_{k,h}=\gamma_{k}^2 + \gamma_{k+h} \gamma_{k-h}$ is an absolutely summable sequence for fixed $h$. The autocovariances geometrically decay for stable Gaussian ARMA, that is there exists a finite $K>0$ and $0<\delta<1$ such that $|\gamma_k| \leq K \delta^{|k|}$. This implies that $M_{k,h}$ is also absolutely summable uniformly in $h$, as by the triangle inequality and geometric decay of $\gamma_h$
\begin{equation*}
    \sup_{h\geq0} |M_{k,h}| \leq |\gamma_k|^2 + |\gamma_{k-h}| | \gamma_{k+h}| \leq K^2 \delta^{2|k|} + K^2 \delta^{|k-h|} \delta^{|k+h|} \leq 2K^2\delta^{2|k|}.
\end{equation*}
Hence by picking for example $C=4K^2(1-\delta^2)^{-1}$, one has $\sup_{h \geq 0} \sum_{k=-\infty}^\infty |M_{k,h}| \leq C$. Now combine the uniform summability result with $\sup_{|k|<T-h}(1-|k|(T-h)^{-1}) \leq 1$ and the the triangle inequality. This delivers
\begin{equation*}
    \sup_{0 \leq h \leq L_T}\mathbb{E}[(\hat{\phi}_h- \gamma_h)^2 ] \leq (T-L_T)^{-1}C.
\end{equation*}
Via Boole's inequality, for any $\varepsilon>0$, one has 
\begin{equation*}
    \mathbb{P} \Big( \sup_{0 \leq h \leq L_T} |\hat{\phi}_h - \gamma_h| > \varepsilon \Big) =  \mathbb{P} \Big( \bigcup_{h=0}^{L_T} |\hat{\phi}_h - \gamma_h| > \varepsilon \Big) \leq \sum_{h=0}^{L_T} \mathbb{P} (  |\hat{\phi}_h - \gamma_h| > \varepsilon ).
\end{equation*}
Chebychev's inequality then delivers
\begin{equation*}
    \mathbb{P} \Big( \sup_{0 \leq h \leq L_T} |\hat{\phi}_h - \gamma_h| > \varepsilon \Big) \leq \sum_{h=0}^{L_T}\frac{\mathbb{E}[(\hat{\phi}_h - \gamma_h)^2]}{\varepsilon^2}.
\end{equation*}
Inserting our previously derived bound, one obtains
\begin{equation*}
     \mathbb{P} \Big( \sup_{0 \leq h \leq L_T} |\hat{\phi}_h - \gamma_h| > \varepsilon \Big) \leq \varepsilon^{-2} (L_T+1) (T-L_T)^{-1} C = \varepsilon^{-2} C \Big( \tfrac{L_T +1}{T} \Big) \Big( \tfrac{1}{1-L_T/T} \Big).
\end{equation*}
This vanishes as $L_T/T \rightarrow 0._{\qedsymbol{}}$

\begin{lemma} \label{lem:restterms}
    Under Assumption \ref{ass:var}, uniformly in $0 \leq h \leq L_T$, one has
    \begin{align*}
       &\beta^{T-2} R_{T,h}^{(1)} = \beta^{T-2} \sum_{t=h+1}^T \epsilon_{t-h} \sum_{i=0}^{h-1}\rho^i \epsilon_{t-i-1}  \overset{\mathbb{P}}{\rightarrow} 0 \\
       &\beta^{T-2} R_{T,h}^{(2)} = \beta^{T-2}  \sum_{t=h+1}^T \epsilon_{t} \sum_{i=0}^{h-1} \beta^{h-i} \epsilon_{t-i-1}  \overset{\mathbb{P}}{\rightarrow} 0 \\
       &\beta^{2(T-2)} R_{T,h}^{(3)} = \beta^{2(T-2)}   \sum_{t=h+1}^T x_{t-1} \sum_{i=0}^{h-1} \beta^{h-i}\epsilon_{t-i-1} \overset{\mathbb{P}}{\rightarrow} 0.
    \end{align*}
\end{lemma}
\noindent \textbf{Proof:}  Under Assumption \ref{ass:var}, there exist finite $M,K$, such that uniformly in $t$, $\mathbb{E}|\epsilon_t^2 | \leq M^2$ and $\mathbb{E} z_t^2 \leq K$, the existence of the latter is guaranteed by Lemma \ref{lem:z_T_bounds}. Via the triangle inequality, then the Cauchy Schwarz bounds $\mathbb{E}|\epsilon_t \epsilon_j| \leq M^2$ or $\mathbb{E}|z_t \epsilon_{j}| \leq \sqrt{M^2 K}$ and bounding finite geometric series by their limits, one may obtain
\begin{align*}
    \mathbb{E}|R_{T,h}^{(1)}| =& \mathbb{E} \Big| \sum_{t=h+1}^T \epsilon_{t-h} \sum_{i=0}^{h-1}\rho^i \epsilon_{t-i-1} \Big| \leq M^2 \sum_{t=h+1}^T \sum_{i=0}^{h-1} |\rho|^i \leq T M^2 \frac{|\rho|^{h-1}}{1-|\beta|} \\
    \mathbb{E}|R_{T,h}^{(2)}| =& \mathbb{E} \Big| \sum_{t=h+1}^T \epsilon_{t} \sum_{i=0}^{h-1} \beta^{h-i} \epsilon_{t-i-1} \Big| \leq M^2 \sum_{t=h+1}^T \sum_{i=0}^{h-1} |\beta|^{h-i} \leq \frac{T M^2}{1-|\beta|} \\
    \mathbb{E}|\beta^{T-2}R_{T,h}^{(3)}| =& \mathbb{E} \Big| \beta^{T-2}\sum_{t=h+1}^T x_{t-1} \sum_{i=0}^{h-1} \beta^{h-i} \epsilon_{t-i-1} \Big| = \mathbb{E} \Big| \sum_{t=h+1}^{T} \beta^{T-t} z_{t}  \sum_{i=0}^{h-1} \beta^{h-i} \epsilon_{t-i-1} \Big| \\
    \leq & \sum_{t=h+1}^T\sum_{i=0}^{h-1} |\beta|^{T-t+h-i} \mathbb{E}|z_t \epsilon_{t-i-1}| \leq \frac{\sqrt{KM^2}}{(1-|\beta|)^2}.\end{align*}
We now scale the previous bounds by a further $\beta^{T-2}$, and sum them across $ 0 \leq h \leq L_T$, delivering
\begin{align*}
    & \sum_{h=0}^{L_T} \mathbb{E} |\beta^{T-2} R_{T,h}^{(1)}| \leq |\beta|^{T-2}  \sum_{h=0}^{L_T} T M^2 \frac{|\rho|^{h-1}}{1-|\beta|} \\
    & \hspace{84pt} \leq  \sum_{h=0}^{L_T} T M^2 \frac{|\beta|^{T-L_T-1}}{1-|\beta|} \leq |\beta|^{T-L_T-1} \frac{T^2 M^2}{1-|\beta|} \\
    & \sum_{h=0}^{L_T} |\beta^{T-2} R_{T,h}^{(2)}| \leq |\beta|^{T-2} \sum_{h=0}^{L_T} \frac{TM^2}{1-|\beta|} \leq |\beta|^{T-2} \frac{T^{2}M^2}{1-|\beta|} \\
    & \sum_{h=0}^{L_T} |\beta^{2(T-2)} R_{T,h}^{(3)}| \leq |\beta|^{T-2} \sum_{h=0}^{L_T} \frac{\sqrt{KM^2}}{(1-|\beta|)^2} \leq |\beta|^{T-2} T \frac{\sqrt{KM^2}}{(1-|\beta|)^2}.
\end{align*}
These vanish as $T \rightarrow \infty$. By Lemma \ref{lem:momtounif}, this implies the uniform convergence we wished to prove$._{\qedsymbol{}}$

\subsubsection{Proof of Lemma \ref{lem:ACFs}} 
Recall that $(\hat{\rho}-\rho)=A_T/B_t$. It follows that
\begin{align*}
    \hat{\gamma}_h =  &  \tfrac{1}{T-h} \bigg[ \sum_{t=h+1}^T \epsilon_t \epsilon_{t-h} - \frac{A_T}{B_T} \sum_{t=h+1}^{T} x_{t-1}\epsilon_{t-h} \\
    & - \frac{A_T}{B_T} \sum_{t=h+1}^{T} x_{t-h-1} \epsilon_t + \Big( \frac{A_T}{B_T} \Big)^2 \sum_{t=h+1}^{T} x_{t-1} x_{t-h-1} \bigg].
\end{align*}
By recursive and forward substitution of $x_t$, recalling the definitions of $A_T, B_T$, one has
\begin{align*}
    \sum_{t=h+1}^T x_{t-1} \epsilon_{t-h} =& \rho^h A_{T-h} + \sum_{t=h+1}^T \epsilon_{t-h} \sum_{i=0}^{h-1}\rho^i \epsilon_{t-i-1} =   \rho^h A_{T-h} + R_{T,h}^{(1)}
    \\
    \sum_{t=h+1}^T x_{t-h-1} \epsilon_t =& \beta^{h} (A_T - A_h) - \sum_{t=h+1}^T \epsilon_{t} \sum_{i=0}^{h-1} \beta^{h-i} \epsilon_{t-i-1} = \beta^{h} (A_T - A_h) - R_{T,h}^{(2)}
    \\
    \sum_{t=h+1}^T x_{t-1} x_{t-h-1} = & \beta^h (B_T - B_h) - \sum_{t=h+1}^T x_{t-1} \sum_{i=0}^{h-1} \beta^{h-i}\epsilon_{t-i-1} = \beta^h (B_T - B_h) - R_{T,h}^{(3)}. 
\end{align*}
The previous identities combined with the triangle inequality and selective multiplication and division by $\beta^{T-2}$ achieves a decomposition of $|\hat{\gamma}_h - \gamma_h|$ of the form
\begin{align*}
    & \sup_{0 \leq h \leq L_T} | \hat{\gamma}_h - \gamma_h| \\
    \leq & \sup_{0 \leq h \leq L_T} \Big| \frac{1}{T-h} \sum_{t=h+1}^T \epsilon_t \epsilon_{t-h} - \gamma_h \Big| \\
    + & \sup_{0 \leq h \leq L_T}  \bigg| \frac{(T-h)^{-1/2} \beta^{T-2} A_T}{\beta^{2(T-2)}B_T} \Big[ (T-h)^{-1/2}\beta^{T-2-h} A_{T-h} + \beta^{T-2} (T-h)^{-1/2} R_{T,h}^{(1)} \Big] \bigg| \\
    + & \sup_{0 \leq h \leq L_T} \bigg|   \frac{(T-h)^{-1/2} \beta^{T-2} A_T}{\beta^{2(T-2)}B_T} \Big[ \beta^{T-2+h}(T-h)^{-1/2}(A_T-A_h) \\ 
     & \hspace{144pt} - (T-h)^{-1/2} \beta^{T-2} R_{T,h}^{(2)} \Big]  \bigg| \\
    + & \sup_{0 \leq h \leq L_T} \bigg| \Big[ \frac{(T-h)^{-1/2} \beta^{T-2} A_T}{\beta^{2(T-2)}B_T} \Big]^2 \Big[ \beta^{2(T-2)+h} (B_T - B_h) - \beta^{2(T-2)} R_{T,h}^{(3)} \Big] \bigg|. 
\end{align*}
The first term converges by Lemma \ref{lem:uniformphi}. The scalings of $R_{T,h}^{(1)}, R_{T,h}^{(2)}, R_{T,h}^{(3)}$ converge uniformly in probability to zero by Lemma \ref{lem:restterms}. As $(T-h)^{-1/2} \leq (T-L_T)^{-1/2}$ is a sequence converging to zero, by Theorem \ref{thm:consistency} the ratio appearing in the second, third and fourth terms converges uniformly to zero as
\begin{equation*}
    \sup_{0 \leq h \leq L_T}\frac{(T-h)^{-1/2} \beta^{T-2} A_T}{\beta^{2(T-2)}B_T} \leq (T-L_T)^{-1/2} \rho^{T-2} (\hat{\rho} - \rho) \overset{\mathbb{P}}{\rightarrow} 0.
\end{equation*}
The terms $(T-h)^{-1/2}\beta^{T-2-h} A_{T-h}$, $\beta^{T-2} A_h, \beta^{2(T-2)}B_h$ vanish asymptotically by Lemma \ref{lem:annoy}, ensuring the second term converges to zero. The third term converge to zero, as recalling $(T-h)^{-1/2} \leq (T-L_T)^{-1/2}$ is a sequence converging to zero, by Lemma \ref{lem:num_conv_zero}, one has
\begin{equation*}
    \sup_{0 \leq h \leq L_T} |\beta^{T-2+h} (T-h)^{-1/2} A_T| \leq (T-L_T)^{-1/2} |\beta^{T-2} A_T| \overset{\mathbb{P}}{\rightarrow} 0.
\end{equation*}
The final object in the fourth term $|\beta^{2(T-2)+h}  B_T| \leq |\beta^{2(T-2)} B_T| \overset{a.s.}{\rightarrow} z^2 (1-\beta^2)^{-1}$ is finite with probability 1, and as multiplied by a ratio converging to zero, the fourth term converges to zero$._{\qedsymbol{}}$

\subsubsection{Proof of Theorem \ref{thm:feasiblearma}}
Via the triangle inequality and supremum bounds, noting $\sup_h|\gamma_h|\leq\gamma_0$, one has
\begin{align*} 
    |\hat{\Gamma} - \Gamma| \leq& |\hat{\gamma}_0 - \gamma_0| + 2 \bigg| \sum_{h=1}^{L_T} \Big\{ \hat{\gamma}_h \hat{\beta}^h - \gamma_h \beta^h\} \bigg| + 2 \bigg| \sum_{h=L_T+1}^\infty \gamma_h \beta^h \bigg| \\
    \leq & |\hat{\gamma}_0 - \gamma_0| + 2 \sum_{h=1}^{L_T} |\hat{\gamma}_h - \gamma_h| |\beta|^{h} + 2 \sum_{h=1}^{L_T} |\gamma_h| |\hat{\beta}^h-\beta^h|  + 2 \frac{\gamma_0 |\beta|^{L_T + 1}}{1-|\beta|} \\
    \leq & \sup_{0 \leq h \leq L_T} |\hat{\gamma}_h-\gamma|  \frac{2}{1-|\beta|} + 2 \gamma_0  \sum_{h=1}^T |\hat{\beta}^h-\beta^h|  + 2 \frac{\gamma_0 |\beta|^{L_T+1}}{1-|\beta|}.
\end{align*}
All the terms of the bound converge in probability to zero: The first vanishes by Lemma \ref{lem:ACFs}, the second by Lemma \ref{lem:uniformbeta} and the third term vanishes as $2 \gamma_0 / (1-|\beta|)$ is a finite constant and $\beta^{L_T} \rightarrow 0.$ To prove the second statement, recall that Proposition \ref{thm:arma} states $\frac{A_T}{\sqrt{B_T}} \overset{D}{\rightarrow} N(0,\Gamma)$. Apply Slutsky's theorem to conclude $\hat{\Gamma}^{-1} \frac{A_T}{\sqrt{B_T}} \overset{D}{\rightarrow} N(0,1)._{\qedsymbol{}}$

\subsection{AR(2) results} \label{sec:ar2proofs}
Recall that in Subsection \ref{sec:ar2}, $x_t$ is the AR(2) process that via $M x_t =: (u_t, 1, w_t)$ can be decomposed into $w_t = \rho w_{t-1} + \epsilon_t$, an explosive AR(1) component, $u_t = \alpha u_{t-1} + \epsilon_t$ a stationary AR(1) component and an intercept term. Throughout Subsection \ref{sec:ar2proofs}, define $\beta := \rho^{-1}$ and $\beta^{T-2}w_{t-1} := z_t$, such that $z_t$ has the same behavior as in earlier sections of the paper. Throughout, we use the entrywise max norm, that is $||A||:=max_{i,j}|a_{i,j}|$; thus a matrix/vector is $o_\mathbb{P}(1)$ if each element converges to zero in probability.

\begin{lemma}\label{lem:H}
Suppose the innovations satisfy Assumptions \ref{ass:var}, \ref{ass:z_0}, \ref{ass:stableAR}. Then $H_T = o_\mathbb{P}(1)$.
\end{lemma}
\noindent \textbf{Proof} Recall that $H_T$ consists of a zero diagonal and correlations $\phi(\cdot, \cdot)$ on the off-diagonal. This proof examines those off-diagonal elements. For $\hat{\phi}(u_{t-1},1)$, Assumption \ref{ass:stableAR} and the continuous mapping theorem immediately imply
\begin{equation*}
    \hat{\phi}(u_{t-1},1) = T^{-1/2} \bigg[T^{-1}\sum_{t=1}^T u_{t-1} \bigg]\bigg[T^{-1}\sum_{t=1}^T u_{t-1}^2 \bigg]^{-1/2} \overset{\mathbb{P}}{\rightarrow} 0.
\end{equation*}
For $\hat{\phi}(u_{t-1},w_{t-1})$, examine the covariance in the numerator $\sum_{t=1}^T w_{t-1} u_{t-1}$. Note that via the identity $z_t = \beta^{t-2}w_{t-1}$
\begin{equation*}
    T^{-1/2} \beta^{T-2} \sum_{t=1}^T w_{t-1} u_{t-1} = T^{-1/2} \sum_{t=1}^T \beta^{T-t} z_t u_{t-1}.
\end{equation*}
Applying the Cauchy-Schwarz inequality and the bounds $\sup_t \mathbb{E}z_t^2 \leq K$ from Lemma \ref{lem:z_T_bounds}, requiring Assumption \ref{ass:var} and $\mathbb{E}[u_t^2]=\sigma_u^2$ from Assumption \ref{ass:stableAR}, one then finds as $T \rightarrow \infty$
\begin{equation*}
    \mathbb{E} \Big| T^{-1/2}\beta^{T-2} \sum_{t=1}^T w_{t-1} u_{t-1} \Big| \leq T^{-1/2}  \sum_{t=1}^T |\beta|^{T-t} \{ \mathbb{E} z_{t}^2 \mathbb{E} u_{t-1}^2 \}^{1/2} \leq T^{-1/2} \frac{\sqrt{ K \sigma_u^2}}{1-|\beta|} \rightarrow 0.
\end{equation*}
Hence $T^{-1/2} \beta^{T-2} \sum_{t=1}^T w_{t-1} u_{t-1} \overset{\mathbb{P}}{\rightarrow} 0$. Recall that by Theorem \ref{thm:and_dist}, requiring Assumption \ref{ass:var}, that $\beta^{2(T-2)}\sum_{t=1}^Tw_{t-1}^2 \overset{a.s.}{\rightarrow} (1-\beta^2)^{-1}z^2$ and that Assumption \ref{ass:z_0} implies $\{z\in \{0\} \}$ is a null set. Furthermore, recall that $T^{-1} \sum_{t=1}^T u_t^2 \overset{\mathbb{P}}{\rightarrow} \sigma_u^2$ by Assumption \ref{ass:stableAR}. It follows that
\begin{equation*}
    \hat{\phi}(u_{t-1},w_{t-1}) = \bigg[ \beta^{T-2}T^{-1/2} \sum_{t=1}^T w_{t-1} u_{t-1} \bigg] \bigg[ \beta^{2(T-2)} \sum_{t=1}^T w_{t-1}^2 T^{-1}\sum_{t=1}^T u_{t-1}^2\bigg]^{-1/2} \overset{\mathbb{P}}{\rightarrow} 0.
\end{equation*}
For $\hat{\phi}(1,w_t) \overset{\mathbb{P}}{\rightarrow} 0$, again examine the covariance in the numerator, recalling the identity $\beta^{t-2} w_{t-1} = z_t$. This yields
\begin{equation*} \textstyle
    T^{-1/2} \beta^{T-2} \sum_{t=1}^T w_{t-1} = T^{-1/2} \sum_{t=1}^T \beta^{T-t} z_t.
\end{equation*}
Recalling $\sup_t \mathbb{E}z_t^2 \leq K$, via the triangle and Cauchy-Schwarz inequalities, as $T \rightarrow \infty$
\begin{equation*} \textstyle
    \mathbb{E} \big| T^{-1/2} \beta^{T-2} \sum_{t=1}^T w_{t-1}| \leq T^{-1/2} \sum_{t=1}^T |\beta|^{T-t} \{\mathbb{E}z_t^2\}^{1/2} \leq T^{-1/2}\sqrt{K}(1-|\beta|)^{-1} \rightarrow 0.
\end{equation*}
Hence, $T^{-1/2} \beta^{T-2} \sum_{t=1}^T w_{t-1} \overset{\mathbb{P}}{\rightarrow} 0$. As before, the denominator of $\hat{\phi}(1,w_{t-1})$ converges by Theorem \ref{thm:and_dist} to $(1-\beta^2)^{-1}z^2$, which is non-zero with probability $1$ by Assumption \ref{ass:z_0}. We therefore find
\begin{equation*}
    \hat{\phi}(1,w_{t-1}) = \Big[ T^{-1/2} \beta^{T-2} \sum_{t=1}^T w_{t-1} \Big] \Big[ \beta^{2(T-2)} \sum_{t=1}^T w_{t-1}^2 \Big]^{1/2} \overset{\mathbb{P}}{\rightarrow} 0._{\qedsymbol{}}
\end{equation*}

\subsubsection{Proof of Proposition \ref{prop:consistency}}
By Assumption \ref{ass:stableAR} and the continuous mapping theorem, one can establish
\begin{equation}\label{eq:AR1-2}
    \frac{T^{-1}\sum_{t=1}^T u_{t-1} \epsilon_t}{ T^{-1} \sum_{t=1}^Tu_{t-1}^2 } \overset{\mathbb{P}}{\rightarrow} \frac{\mathbb{E}[u_{t-1} \epsilon_t]}{\mathbb{E}[u_{t-1}^2]} = 0, \quad \frac{\sum_{t=1}^T \epsilon_t}{T} \overset{\mathbb{P}}{\rightarrow} \mathbb{E}[\epsilon_t] =0.
\end{equation}
Next, recall that $w_t$ is an explosive AR(1) with innovations satisfying Assumptions \ref{ass:var}-\ref{ass:z_0}. Hence, Corollary \ref{cor:consistency} applies. This delivers
\begin{equation} \label{eq:ar3}
    \frac{\sum_{t=1}^T w_{t-1} \epsilon_t}{\sum_{t=1}^T w_{t-1}^2} \overset{a.s.}{\rightarrow} 0.
\end{equation}
As Assumptions \ref{ass:var}, \ref{ass:z_0}, \ref{ass:stableAR} hold, Lemma \ref{lem:H} applies and $H_T=o_\mathbb{P}(1)$. We therefore have that
\begin{equation*}
    (\hat{\theta} - \theta) = M' (Y'Y)^{-1}Y'\epsilon = M'   \left[ D_T (I_3 + o_\mathbb{P}(1)) D_T\right]^{-1}  \sum_{t=1}^T \begin{pmatrix}
        u_{t-1} \epsilon_t \\
        \epsilon_t \\
        w_{t-1} \epsilon_t
    \end{pmatrix}.
\end{equation*}
By $(\ref{eq:AR1-2})$-$(\ref{eq:ar3})$, requiring Assumptions \ref{ass:var}-\ref{ass:z_0} and \ref{ass:stableAR} it follows that
\begin{equation*}
    (\hat{\theta}-\theta) = M' (I_3+o_\mathbb{P}(1))^{-1} \begin{pmatrix}
        \sum_{t=1}^T u_{t-1} \epsilon_t / \sum_{t=1}^T u_{t-1}^2 \\
        \sum_{t=1}^T\epsilon_t / T \\
        \sum_{t=1}^T w_{t-1} \epsilon_t / \sum_{t=1}^T w_{t-1}^2
    \end{pmatrix} \overset{\mathbb{P}}{\rightarrow} 0._{\qedsymbol{ }}
\end{equation*}

\newpage
\printbibliography
\end{document}